\documentclass[11pt,a4paper]{article}
\usepackage{amsfonts,amsthm,amsmath,amssymb}
\usepackage{comment}
\usepackage{mathtools}
\usepackage{subcaption}
\usepackage{url}
\usepackage{multirow}
\usepackage{wrapfig}
\usepackage[pdftex,dvipsnames]{xcolor}  
\allowdisplaybreaks

\newtheorem{theorem}{Theorem}
\newtheorem{proposition}[theorem]{Proposition}
\newtheorem{lemma}[theorem]{Lemma}
\newtheorem{corollary}[theorem]{Corollary}
\newtheorem{definition}[theorem]{Definition}

\newtheorem{observation}[theorem]{Observation}

\usepackage{authblk}

\begin{document}

\title{Approximating shortest paths in weighted square and hexagonal meshes}

\date{}

\author[1]{Prosenjit Bose\thanks{Email: jit@scs.carleton.ca}}
\author[1,2]{Guillermo Esteban\thanks{Email: g.esteban@uah.es}}
\author[2]{David Orden\thanks{Email: david.orden@uah.es}}
\author[3]{Rodrigo I. Silveira\thanks{Email: rodrigo.silveira@upc.edu}}

\affil[1]{School of Computer Science, Carleton University, Canada}
\affil[2]{Departamento de F\'{i}sica y Matem\'{a}ticas, Universidad de Alcal\'{a}, Spain}
\affil[3]{Departament de Matemàtiques, Universitat Politècnica de Catalunya, Spain}

\maketitle

\begin{abstract}
    Continuous 2-dimensional space is often discretized by considering a mesh of weighted cells. In this work we study how well a weighted mesh approximates the space, with respect to shortest paths. We consider a shortest path~$ \mathit{SP_w}(s,t) $ from~$ s $ to $ t $ in the continuous 2-dimensional space, a shortest vertex path~$ \mathit{SVP_w}(s,t) $ (or any-angle path), which is a shortest path where the vertices of the path are vertices of the mesh, and a shortest grid path $ \mathit{SGP_w}(s,t) $, which is a shortest path in a graph associated to the weighted mesh. We provide upper and lower bounds on the ratios $ \frac{\lVert \mathit{SGP_w}(s,t)\rVert}{\lVert \mathit{SP_w}(s,t)\rVert} $, $ \frac{\lVert \mathit{SVP_w}(s,t)\rVert}{\lVert \mathit{SP_w}(s,t)\rVert} $, $ \frac{\lVert \mathit{SGP_w}(s,t)\rVert}{\lVert \mathit{SVP_w}(s,t)\rVert} $ in square and hexagonal meshes, extending previous results for triangular grids. These ratios determine the effectiveness of existing algorithms that compute shortest paths on the graphs obtained from the grids. Our main results are that the ratio $ \frac{\lVert \mathit{SGP_w}(s,t)\rVert}{\lVert \mathit{SP_w}(s,t)\rVert} $ is at most $ \frac{2}{\sqrt{2+\sqrt{2}}} \approx 1.08 $ and $ \frac{2}{\sqrt{2+\sqrt{3}}} \approx 1.04 $ in a square and a hexagonal mesh, respectively.
\end{abstract}

\section{Introduction}\label{cap.introduction}

Computing a shortest path between two points $ s $ and $ t $ plays an important role in applications like isotropic sampling~\cite{fu2009direct,ying2014parallel}, image registration~\cite{paragios2003non,shekhovtsov2008efficient}, shape characterization~\cite{reuter2006laplace}, texture mapping~\cite{zhang2003synthesis}, and image segmentation~\cite{buyssens2014eikonal,wang2017superpixel}, among others. Shortest path problems can be categorized by various factors, which include the domain (e.g., discrete, continuous), the distance measure used (e.g., geodesic distance, link-distance, or Euclidean metric), or the
type and the number of domain constraints (e.g., holes in polygons, or obstacles in the plane), see~\cite{giannelli2016path,kapoor1997efficient,ramanathan2013shortest,wang2017discrete}.

Real-time applications where shortest paths are used, like geographic information systems~\cite{floriani}, robotics~\cite{gaw,rowe,Sharir}, or gaming~\cite{kamphuis,sturtevant2}, feature increasingly large amounts of information. All these data are expected to be managed efficiently in terms of execution time and solution quality. So, one alternative that is often encountered in many graphic applications to reduce time and space complexity is to discretize the 2D space by considering \emph{tessellations}. These are partitions of the space into simple polygons, with the property that the intersection of any of these polygons is a point, a segment, or the empty set. Convex polygons, and overlapping disks ---of different radii--- are among the most frequently used region shapes~\cite{van2016comparative}. Nonetheless, \emph{regular grids} are the dominant and simplest shape representation scheme in computer graphics~\cite{carsten2009global,jiang2009interpolatory,losasso2004simulating,Minecraft}. In general, regular meshes are easy to generate, require less memory (array storage can define neighbor connectivity implicitly), and are a natural choice for environments that are grid-based by design (e.g., finite element analysis, modeling, and many game designs, see Figure~\ref{fig:games}). In addition, some search and optimization algorithms can be optimized for square~\cite{Ammar,nagy} or hexagonal grids~\cite{bjornsson2003comparison}.

\begin{figure}[tb]
    \begin{subfigure}[t]{0.5\textwidth}
        \centering
        \includegraphics[scale=0.465]{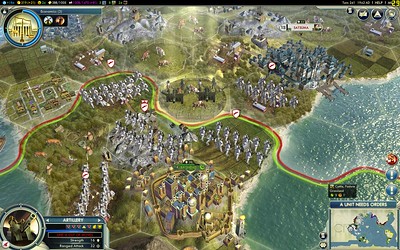}
        \caption{``CivilizationV\_DX11 2010-11-08 22-06-24-30'' by Chris Moore, used under CC BY 2.0.}
    \end{subfigure}
    \qquad\qquad
    \begin{subfigure}[t]{0.4\textwidth}
        \centering
        \includegraphics[scale=0.15]{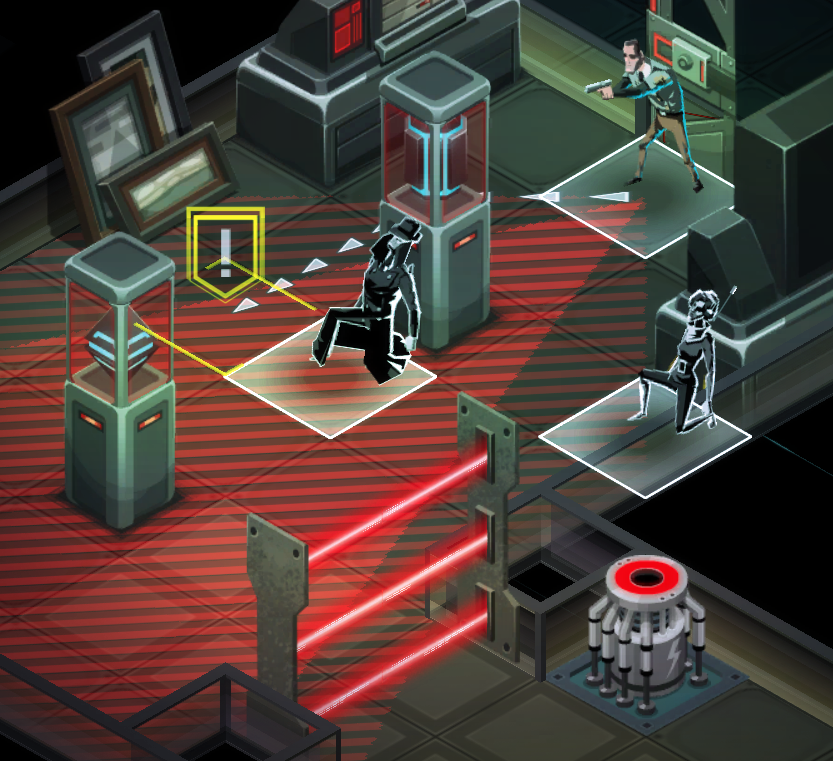}
        \caption{Screenshot of the ``Invisible, Inc.'' video game by Klei.}
        \label{fig:invisible}
    \end{subfigure}
    \caption{Note the regular meshes underlying the scenes.}
    \label{fig:games}
\end{figure}

Continuous 2D environments can be tessellated by using only three types of regular polygons, that are triangles, squares and hexagons. Using triangular meshes over other types of regular meshes offers some advantages: hexagonal cells can be obtained with six triangular cells sharing a common vertex, and the distance from one vertex to any of its six neighbors is always the same, which facilitates the calculation of distances~\cite{ColossalCitadels}. The structure of the square meshes can be exploited to efficiently create bounding volume hierarchies~\cite{kim2011coons}. Also, the axes of movement are orthogonal to one another. In terms of video games, this implies that the movement of non-player characters is free on one axis without affecting the position on the other axis. However, the higher number of sides of a hexagonal cell allows to represent curves in the patterns of the data more naturally than in a square mesh~\cite{hexagonsproperties}. In square and hexagonal mesh, each vertex is connected to multiple vertices within the same cell, which increases the number of possible movements from one vertex. In addition, the regularity of the meshes can also be regarded as a form of parametrization, which further means that it can be mapped nicely onto Cartesian coordinates, accelerating many geometric operations~\cite{cartesian}.

Consider a type of video game where the player primarily uses stealth to overcome or avoid enemies. The opponents typically have a line of sight which the player can avoid by hiding behind objects, or staying in the shadows. See Figure~\ref{fig:invisible} for an example where there is a single enemy observer whose view is obstructed by tall obstacles. These regions where the player remains undetected are considered as zero-cost regions, and those visible to at least one enemy are considered as regions with a fixed cost per unit distance. The problem is to find an obstacle-avoiding path between two locations that minimizes the total amount of time in the exposed regions.

Furthermore, we can consider the more general terrain navigation problem where the goal is to compute a shortest path when the cost of traversing the domain varies depending on the region. In this setting, the interior of each region $ R_i $ of the tessellation has a (non-negative) weight $ \omega_i $, which specifies the cost per unit distance of traveling interior to region $ R_i $. That is, the domain consists of a weighted planar polygonal subdivision. This problem is known as the \emph{weighted region problem} (WRP)~\cite{Mitchell2}.

The WRP is very general, since it allows to model many well-known variants of geometric shortest path problems \cite{de1997trekking,Mitchell1}. However, for many applications, the underlying mesh may be associated with non-uniform density~\cite{shen2016converting}, anisotropic metric~\cite{alliez2003anisotropic,galin2010procedural,kovacs2011anisotropic}, or user-specified geometric constraints~\cite{mitra2014structure,zhang2019real}.

In the WRP, the notion of length is defined according to the \emph{weighted region metric}. A segment~$ \pi $ between two points on the same region has cost $ \omega_i \rVert \pi \lVert $ when traversing the interior of $ R_i $ and $ \min\{\omega_i,\omega_j\}\rVert \pi \lVert $ when lying on the edge between~$ R_i $ and $ R_j $. Thus, the weighted length of a path $ \Pi $ through a tessellation is the sum of the weighted lengths of its subpaths through each face and along each edge. With a slight abuse of notation, this can be denoted by~$ \lVert \Pi \rVert$.

Several algorithms have been suggested for determining shortest paths in a weighted subdivision. However, only few special cases can be solved efficiently, e.g., with distances measured based on $L_1$ metric and each edge of the obstacles parallel to the $x$- or $y$-axis~\cite{ChenKT00,lee1990shortest}, the maximum concealment problem~\cite{GewaliMMN90}, or large grid environments~\cite{Ammar,jigang2010notice,peyer2009generalization}. In more general settings, existing algorithms for the WRP are unpleasing, since they are fairly complex in design and implementation, numerically unstable and may require an exponential number of bits to perform certain computations~\cite{Aleksandrov,Aleksandrov2,Aleksandrov3}.

Recently, it has been proven that computing an exact shortest path between two points using the weighted region metric is an unsolvable problem in the Algebraic Computation Model over the Rational Numbers ($ \mathit{ACM\mathbb{Q}} $)~\cite{Lou}. In the $ \mathit{ACM\mathbb{Q}} $ one can exactly compute any number that can be obtained from rational numbers by applying a finite number of operations from $ +, -, \times, \div, \sqrt[k]{} $, for any integer $ k \geq 2 $. This justifies the search for approximate solutions as opposed to exact ones. These approximate solutions partition each edge of the subdivision into intervals, obtaining paths with length at most $ (1+\varepsilon)$ times the length of the optimal path, for $ \varepsilon \in (0,1] $.

Mitchell and Papadimitriou~\cite{Mitchell2} applied to the WRP a variant of continuous Dijkstra's method by exploiting the fact that shortest paths obey Snell's law of refraction at the boundaries of the regions.

Other researchers, such as Aleksandrov et al.~\cite{Aleksandrov2, Aleksandrov3}, Cheng~\cite{ChengJV15}, or Reif and Sun~\cite{Sun}, worked by computing a discretization of the domain by carefully placing Steiner points either on the boundary, or in the interior of the cells.

Moreover, Rowe and Ross~\cite{rowe}, and Lanthier et al.~\cite{lanthier1999shortest} solved a discretized Hamilton-Jacobi-Bellmann equation and proposed an efficient algorithm to solve an isotropic shortest path problem.

Thus, in practice, exact shortest paths, denoted $\mathit{SP_w}(s,t)$, are not computed: instead, two types of shortest paths that result from different meshes and different definitions of the neighbors of a vertex are usually considered~\cite{bailey2021path}. These two paths are used in practice to approximate shortest paths. However, the approximation error is not fully understood, and in this work we improve upon this, giving theoretical foundations that explain why meshes work so well for this in practice. We recently studied the approximation error when the space is covered with equilateral triangles~\cite{bose2023approximating}. Hence, from now on, we just focus on the other two types of meshes: square and hexagonal.

We define the \emph{$k$-corner grid graphs} ($ G_{k\text{corner}} $) as the graphs whose vertex set is the set of corners of the tessellation, and each vertex is connected by an edge to a set of $ k $ adjacent vertices. The graphs resulting from joining vertices with the edges of the cells are denoted $ G_{4\text{corner}} $ in a square mesh, and $ G_{3\text{corner}} $ in a hexagonal mesh. See Figures~\ref{fig:4corner} and \ref{fig:3corner}, respectively. If we define edges from one vertex to any other vertex of the same cell, the graph that is obtained is denoted $ G_{8\text{corner}} $ in a square mesh, and $ G_{12\text{corner}} $ in a hexagonal mesh. See Figures~\ref{fig:8corner} and \ref{fig:12corner}, respectively.

\begin{figure}[tb]
	    \captionsetup[sub]{justification=centering}
		\centering
     	\begin{subfigure}[b]{0.4\textwidth}
     	\centering
        	\includegraphics{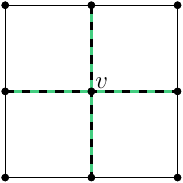}
         	\caption{}
         	\label{fig:4corner}
     	\end{subfigure}
     	\qquad\qquad
     	\begin{subfigure}[b]{0.4\textwidth}
     	\centering
        	\includegraphics{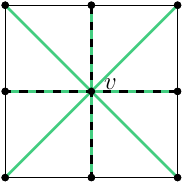}
         	\caption{}
         	\label{fig:8corner}
     	\end{subfigure}
      \qquad\qquad
     	\begin{subfigure}[b]{0.4\textwidth}
     	\centering
        	\includegraphics{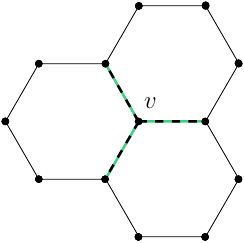}
         	\caption{}
         	\label{fig:3corner}
     	\end{subfigure}
      \qquad\qquad
     	\begin{subfigure}[b]{0.4\textwidth}
     	\centering
        	\includegraphics{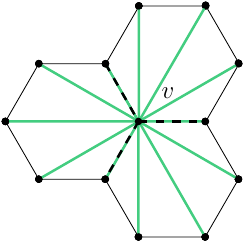}
         	\caption{}
         	\label{fig:12corner}
     	\end{subfigure}
        \caption{Vertex $ v $ is connected to its neighbors in $ G_{4\text{corner}} $ (top left), in $ G_{8\text{corner}} $ (top right), in $ G_{3\text{corner}} $ (bottom left), and in $ G_{12\text{corner}} $ (bottom right). The edges of each graph that coincide with the cell edges are depicted with dashed lines.}
       	\label{fig:6neigh}
	\end{figure}

In addition, we define the \emph{corner-vertex graph} ($ G_{\text{corner}} $) as the graph whose vertex set is the set of corners of the tessellation, and every pair of vertices is connected by an edge. This graph is the complete graph over the set of vertices. Figures~\ref{fig:squarevertex} and \ref{fig:hexvertex} depict the edges from a vertex $ v $ to some of its neighbors in the corner-vertex graph of a square, and a hexagonal mesh, respectively. Note that in these graphs some edges overlap.

\begin{figure}[tb]
	    \captionsetup[sub]{justification=centering}
		\centering
     	\begin{subfigure}[b]{0.44\textwidth}
     	\centering
        	\includegraphics{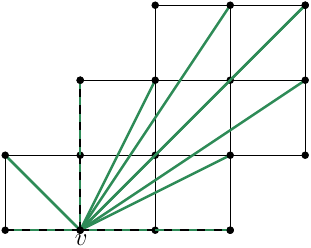}
         	\caption{}
         	\label{fig:squarevertex}
     	\end{subfigure}
     	\qquad
     	\begin{subfigure}[b]{0.49\textwidth}
     	\centering
        	\includegraphics{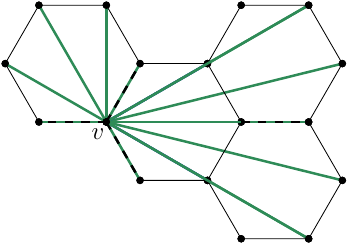}
         	\caption{}
         	\label{fig:hexvertex}
     	\end{subfigure}
        \caption{Vertex $ v $ is connected to some of its neighbors in $ G_{\text{corner}} $ of a square (left) and a hexagonal (right) mesh. The edges of each graph that coincide with the cell edges are depicted with dashed lines.}
       	\label{fig:vertex-graph}
	\end{figure}

A path in $ G_{k\text{corner}} $ is called a \emph{grid path}; a shortest grid path between two vertices $s$ and $t$ will be denoted $\mathit{SGP_w}(s,t)$, where the subscript $w$ highlights that this path depends on a particular weight assignment $w$. A path in $ G_{\text{corner}} $ is called a \emph{vertex path}; a shortest vertex path from vertex $ s $ to vertex $ t $ will be denoted by~$\mathit{SVP_w}(s,t)$.See Figure~\ref{fig:comparison-corner} for a comparison between the three types of shortest paths in $ G_{4\text{corner}} $, where $ \mathit{SP_w}(s,t) $, $ \mathit{SVP_w}(s,t) $, and $\mathit{SGP_w}(s,t)$ are represented in blue, green, and red, respectively. Note that in all figures where the complete paths from $ s $ to $ t $ are represented, cells that are not depicted are considered to have infinite weight. Moreover, two paths that coincide in a segment are represented with dashed lines in that segment.

	\begin{figure}[tb]
		\centering
    	\includegraphics{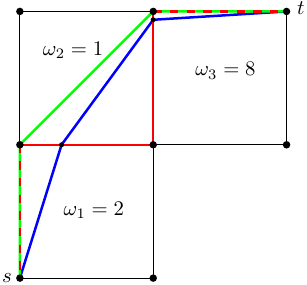}
        \captionof{figure}{The weighted length of $ \mathit{SP_w}(s,t) $ (blue), $ \mathit{SVP_w}(s,t) $ (green), and $ \mathit{SGP_w}(s,t) $ (red) is, respectively, $ 11.26 $, $ 11.41 $, and $ 12 $, when the length of the cell side is $ 1 $ in $ G_{4\text{corner}} $.}
        \label{fig:comparison-corner}
	\end{figure}
	
We stated above that path planning algorithms are faster and more efficient when the 2D space is tessellated. On the other hand, the constrained shortest paths $ \mathit{SVP_w}(s,t) $ and $ \mathit{SGP_w}(s,t) $ are approximations to the actual shortest path~$ \mathit{SP_w}(s,t) $. The quality of such approximations is determined by upper-bounding the ratios $ \frac{\lVert \mathit{SGP_w}(s,t)\rVert}{\lVert \mathit{SP_w}(s,t)\rVert} $ and $ \frac{\lVert \mathit{SVP_w}(s,t)\rVert}{\lVert \mathit{SP_w}(s,t)\rVert} $. In addition, these ratios are useful for designers when choosing path planning algorithms for their applications. The ratio $ \frac{\lVert \mathit{SGP_w}(s,t)\rVert}{\lVert \mathit{SVP_w}(s,t)\rVert} $ is also studied, to see how different the two approximations can be. In most of this paper we focus on $ G_{8\text{corner}} $, which is the most natural graph defined in a square grid. In addition, we discuss how to obtain lower and upper bounds for the same ratios for square $ G_{4\text{corner}} $, and for hexagonal $ G_{3\text{corner}} $ and $ G_{12\text{corner}} $.

\subsection{Previous results}

The comparison of the length of different types of shortest paths in the tessellated space has been previously studied. For instance, Bailey et al.~\cite{bailey2015path} performed an analysis of the three ratios on a square grid. However, they only allowed two weights for the cells: weight $ 1 $ (free space), and weight $ \infty$ (obstacles). In addition, Bailey et al.~\cite{bailey2021path} recently extended the results in~\cite{bailey2015path} by upper-bounding the worst case ratio~$ R= \frac{\lVert \mathit{SGP_w}(s,t)\rVert}{\lVert \mathit{SP_w}(s,t)\rVert} $ on triangular and hexagonal grids when the cells in the tessellated space have weights in the set $ \{1, \infty\} $. They proved that the length of $ \mathit{SGP_w}(s,t) $ in hexagonal $ G_{3\text{corner}} $ and $ G_{12\text{corner}} $, square $ G_{4\text{corner}} $ and $ G_{8\text{corner}} $, and triangle $ G_{6\text{corner}} $ can be up to $ \approx\!1.5 $, $ \approx\!1.04 $, $ \approx\!1.41 $, $ \approx\!1.08 $, and $ \approx\!1.15 $ times the length of~$ \mathit{SP_w}(s,t)$, respectively. In the more general case, when a general weight assignment $ \{1, \omega_i, \ldots, \omega_j, \infty\} $ is given in the WRP, Jaklin~\cite{Bound3} presented an upper bound of $ R \leq 2\sqrt{2} $. However, this ratio is obtained when the vertices of the associated graph are placed at the center of the cells in a square mesh, which is a different setting than placing the vertices on the corners of the cells. For instance, the ratio $ R $ in Figure~\ref{fig:infinite} is unbounded if we consider each vertex of the graph connected to four adjacent vertices.

 \begin{figure}[tb]
                \centering
	            \includegraphics{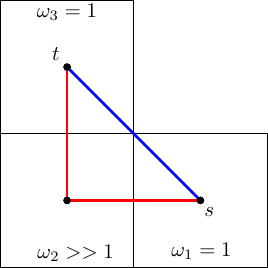}
	            \caption{$ \mathit{SGP_w}(s,t) $ (in red) intersects a cell that $ \mathit{SP_w}(s,t) = \mathit{SVP_w}(s,t) $ (in blue) does not intersect.}
	            \label{fig:infinite}
            \end{figure}

Moreover, we recently analyzed the case of a general weight assignment when considering vertices at corners of triangular cells~\cite{bose2023approximating}. In particular, we upper-bounded the ratio $ R $ by $ \frac{2}{\sqrt{3}} \approx 1.15 $ for equilateral triangle tessellations.

\subsection{Our results}

The main contribution of this paper is the generalization of the results from Bailey et al.~\cite{bailey2021path} when the weight $ \omega_i $ of each region $ R_i $ takes any value in $ \mathbb{R}_{\geq 0} $, for square and hexagonal meshes.

Our main results are that $ R \leq \frac{2}{\sqrt{2+\sqrt{2}}}\approx 1.08 $ when any (non-negative) weight is assigned to the cells of a square mesh and each vertex is connected to its $ 8$ neighboring vertices, and $ R \leq \frac{2}{\sqrt{2+\sqrt{3}}}\approx 1.04 $ when each vertex is connected to its $ 12$ neighboring vertices in a hexagonal mesh. Using these results, we can provide upper bounds on the other ratios. In particular, by applying an analogous procedure as in the graphs $ G_{8\text{corner}} $ and $ G_{12\text{corner}} $ we obtain upper bounds for the ratios in a square $ G_{4\text{corner}} $, and in hexagonal $ G_{3\text{corner}} $. The lower bounds for most of these ratios are given by the results from~\cite{bailey2021path}. Tables~\ref{tab:square} and~\ref{tab:hexw} summarize our results, together with the previously known lower bounds.

\begin{table}[htb]
\centering
\resizebox{0.9\textwidth}{!}{\begin{tabular}{||c|c|c|c||}
\cline{3-4}
\multicolumn{2}{c|}{}                                                                                          & $ G_{4\text{corner}} $ & $ G_{8\text{corner}} $ \\ \hline
\multirow{2}{*}{$ \frac{\lVert \mathit{SGP_w}(s,t)\rVert}{\lVert \mathit{SP_w}(s,t)\rVert} $}  & Lower bound &  $ \sqrt{2} \approx 1.41 $~\cite{bailey2021path} & $ \frac{2}{\sqrt{2+\sqrt{2}}} \approx 1.08 $~\cite{bailey2021path}                    \\ \cline{2-4}
                                                                                               & Upper bound  & $ \sqrt{2} \approx 1.41 $ (Thm.~\ref{thm:4square}) & $ \frac{2}{\sqrt{2+\sqrt{2}}} \approx 1.08 $ (Thm.~\ref{thm:4})                       \\ \hline
\multirow{2}{*}{$ \frac{\lVert \mathit{SVP_w}(s,t)\rVert}{\lVert \mathit{SP_w}(s,t)\rVert} $}  & Lower bound & \multicolumn{2}{c||}{                       $ \frac{\sqrt{2}\sqrt{\sqrt{2}-1}}{(\sqrt{2}-1)^{\frac{3}{2}}-\sqrt{2}+2} \approx 1.07 $ (Obs.~\ref{obs:5})} \\ \cline{2-4}
                                                                                               & Upper bound  & \multicolumn{2}{c||}{$ \frac{2}{\sqrt{2+\sqrt{2}}} \approx 1.08 $ (Cor.~\ref{cor:5})} \\ \hline
\multirow{2}{*}{$ \frac{\lVert \mathit{SGP_w}(s,t)\rVert}{\lVert \mathit{SVP_w}(s,t)\rVert} $} & Lower bound  & $ \sqrt{2} \approx 1.41 $~\cite{bailey2021path} & $ \frac{2}{\sqrt{2+\sqrt{2}}} \approx 1.08 $~\cite{bailey2021path} \\ \cline{2-4}
                                                                                               & Upper bound & $ \sqrt{2} \approx 1.41 $ (Cor.~\ref{cor:sgpsvp4}) & $ \frac{2}{\sqrt{2+\sqrt{2}}} \approx 1.08 $ (Cor.~\ref{cor:3})         \\ \hline             
\end{tabular}}
\captionof{table}{Results obtained for square meshes.}
		\label{tab:square}
\end{table}

\begin{table}[htb]
\centering
\resizebox{0.9\textwidth}{!}{\begin{tabular}{||c|c|c|c||}
\cline{3-4}
\multicolumn{2}{c|}{}                                                                                         & $ G_{3\text{corner}} $ & $ G_{12\text{corner}} $  \\ \hline
\multirow{2}{*}{$ \frac{\lVert \mathit{SGP_w}(s,t)\rVert}{\lVert \mathit{SP_w}(s,t)\rVert} $}  & Lower bound & $ \frac{3}{2} = 1.5 $~\cite{bailey2021path} & $ \frac{2}{\sqrt{2+\sqrt{3}}} \approx 1.04 $~\cite{bailey2021path} \\ \cline{2-4}
                                                                                               & Upper bound & $ \frac{3}{2} = 1.5 $ (Thm.~\ref{thm:ratiohex3corner})                       & $ \frac{2}{\sqrt{2+\sqrt{3}}} \approx 1.04 $ (Thm.~\ref{thm:ratiohex12corner})                 \\ \hline
\multirow{2}{*}{$ \frac{\lVert \mathit{SVP_w}(s,t)\rVert}{\lVert \mathit{SP_w}(s,t)\rVert} $}  & Lower bound & \multicolumn{2}{c||}{$ \frac{2\sqrt{4\sqrt{3}-6}}{(2-\sqrt{3})(\sqrt{4\sqrt{3}-6}+6)} \approx 1.03 $ (Obs.~\ref{obs:lowerhex})} \\ \cline{2-4}
                                                                                               & Upper bound & \multicolumn{2}{c||}{$ \frac{2}{\sqrt{2+\sqrt{3}}} \approx 1.04 $ (Cor.~\ref{cor:7})}  \\ \hline
\multirow{2}{*}{$ \frac{\lVert \mathit{SGP_w}(s,t)\rVert}{\lVert \mathit{SVP_w}(s,t)\rVert} $} & Lower bound & $ \frac{3}{2} = 1.5 $~\cite{bailey2021path} & $ \frac{2}{\sqrt{2+\sqrt{3}}} \approx 1.04 $~\cite{bailey2021path} \\ \cline{2-4}
                                                                                               & Upper bound & $ \frac{3}{2} = 1.5 $ (Cor.~\ref{cor:ratiohex3corner})                       & $ \frac{2}{\sqrt{2+\sqrt{3}}} \approx 1.04 $ (Cor.~\ref{cor:ratiohex12corner})             \\ \hline             
\end{tabular}}
\captionof{table}{Results obtained for hexagonal meshes.}
		\label{tab:hexw}
\end{table}

\subsection{Comparison with triangle case}

The analysis we used in~\cite{bose2023approximating} was the following: we first defined a grid path, which we called a \emph{crossing path $X(s,t)$}, that intersects the boundary of the same set of cells intersected by $ \mathit{SP_w}(s,t) $. The objective was to upper-bound the ratio between the weighted length of the two paths between consecutive points where they coincide by using that the union of those two paths generates six types of simple polygons. For one of those types, we needed to further refine the analysis by introducing another class of grid paths, called \emph{shortcut paths $ \mathit{\Pi}(s,t)$}, that intersect the edge of the cell containing the polygon that was not intersected by the crossing path.

The approach we use here is similar to the one used in \cite{bose2023approximating} for triangular tessellations. However, for these two different types of tessellations we need a more refined analysis to address several cases that do not appear in a triangular tessellation. In particular, two consecutive points where the crossing path $ X(s,t) $ and $ \mathit{SP_w}(s,t) $ coincide, might not belong to adjacent cells. In the case of square cells, we obtain two of these cases; and in the case of hexagonal cells, three such cases. Similarly to the triangular case, we need to define certain types of shortcut paths that help us to improve the upper bounds on the ratios. Indeed, for one type of simple polygons in a hexagonal mesh we need more than one shortcut path intersecting it, so that we can obtain the desired ratio. This requires more intricate arguments than those in~\cite{bose2023approximating} when comparing the lengths of the grid paths that intersect certain types of polygons. Nonetheless, we overcame these difficulties, and derived bounds for \emph{any} weight assignment for the regions. What was surprising is that although the model where arbitrary weights can be assigned to the cells is much more flexible, we prove that the worst-case upper bounds on the ratios defined above are independent of the weights assigned to the regions in the tessellations. Although counter-intuitive, in this paper we show that the worst-case occurs when cells all have either infinite weight or the same fixed constant weight.

\section{Square meshes}\label{cap.sq}
	
	In Section \ref{cap.introduction}, we introduced some bounds previously existing for the worst-case ratios between the length of the three types of shortest paths defined on a square mesh. These ratios were calculated when the weight of a cell is either~$ 1 $ or $ \infty $, or when the vertices of the graph are placed at the center of the cells, whose weight is any non-negative real value. In this work, we obtain new results by considering weights that can take any non-negative real value, and for the case where the vertices are placed at the corners of the cells. Without loss of generality, we suppose the length of each side of the square cells is~$ 1 $, so that the cell diagonal has length $ \sqrt{2} $. Let~$ \mathit{SP_w}(s, t) $ be a weighted shortest path in a square mesh $ \mathcal{S} $ from a corner $ s $ to another corner $ t $.

	\subsection{$ \frac{\lVert \mathit{SGP_w}(s,t)\rVert}{\lVert \mathit{SP_w}(s,t)\rVert} $ ratio in $ G_{8\text{corner}} $}
	\label{sec:8sq}

	This section is dedicated to obtaining, for two corners $ s $ and~$ t $, an upper bound on the ratio~$ \frac{\lVert \mathit{SGP_w}(s,t)\rVert}{\lVert \mathit{SP_w}(s,t)\rVert} $ in~$ G_{8\text{corner}} $ for a square mesh~$\mathcal{S}$ where the faces are closed squares that are assigned arbitrary weights in $ \mathbb{R}_{\geq 0} $. 
	
	\subsubsection{Crossing path and simple polygons}
	
	Knowing the exact shape of the shortest paths~$ \mathit{SP_w}(s,t) $ is difficult. In addition, the set of cells intersected by $ \mathit{SP_w}(s,t) $ and $ \mathit{SGP_w}(s,t) $ might be different, which implies considering a different set of weights for each of the paths. To address all these issues in a triangular tessellation, in~\cite{bose2023approximating} we defined a grid path $ X(s,t) $ called a \textit{crossing path} for a given $\mathit{SP_w}(s,t)$. One of the most important properties of the crossing path is that it only intersects the edges of the cells intersected by $\mathit{SP_w}(s,t)$. Moreover, $ X(s,t) $ is a grid path, so $ \lVert \mathit{SGP_w}(s,t) \rVert \leq \lVert X(s,t) \rVert $, thus the key idea to prove an upper bound on the ratio $ \frac{\lVert \mathit{SGP_w}(s,t)\rVert}{\lVert \mathit{SP_w}(s,t)\rVert} $ is to upper-bound the ratio $ \frac{\lVert \mathit{X}(s,t)\rVert}{\lVert \mathit{SP_w}(s,t)\rVert} $ instead. When the space is discretized using a square mesh we can analogously apply this idea of defining a crossing path $ X(s,t) $ for a given $\mathit{SP_w}(s,t)$, see the orange path in Figure~\ref{fig:crossingsquare}. $ X(s,t)$ will be defined based on the different positions of the points where $ \mathit{SP_w}(s,t) $ intersects the edges of a cell. We assume that $ \mathit{SP_w}(s,t) $ is unique, otherwise it is enough to define a crossing path for each shortest path and repeat the arguments where we compute upper bounds on the ratios in the upcoming sections.
	
	\begin{figure}[tb]
		\centering
		\includegraphics[width=\textwidth]{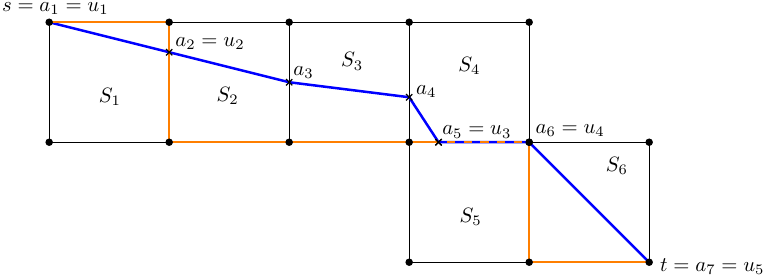}
		\captionof{figure}{Weighted shortest path $ \mathit{SP_w}(s,t) $ (blue) and the crossing path $ X(s,t) $ (orange) from~$ s $ to $ t $ in a square mesh.}
		\label{fig:crossingsquare}
	\end{figure}
	
	Let $ (S_1, \ldots, S_n) $ be the ordered sequence of consecutive square cells intersected by~$ \mathit{SP_w}(s, t) $ in the mesh $ \mathcal{S} $. We remind the reader that we never actually compute ~$ \mathit{SP_w}(s, t) $. We prove the bounds by studying geometric properties of shortest paths and essentially using those properties to define a grid path that is easier to bound.
 Let $ (s = a_1, a_2, \ldots, a_{n+1} = t) $ be the sequence of consecutive points where~$ \mathit{SP_w}(s, t) $ changes the cell(s) it belongs to in $ \mathcal{S} $. In particular, let~$ a_i $ and~$ a_{i+1} $ be, respectively, the points where $ \mathit{SP_w}(s,t) $ enters and leaves $ S_i $. 
	
	\begin{definition}
		\label{def:crossingsquare}
		The crossing path $ X(s,t) $ between two vertices $ s $ and $ t $ in a square mesh~$ \mathcal{S} $ is the path in $ \mathcal{S} $ with vertex sequence $ (X_1, \ldots, X_n) $, where $ X_i $ is a sequence of at most three vertices determined by the pair~$ (a_i, a_{i+1}), \ 1 \leq i \leq n $. Let $ e^i_1 \in S_i $ be an edge containing~$ a_{i} $, then:
		\begin{enumerate}
		\setlength\itemsep{0em}
		    
		    \item If $ a_i $ and $ a_{i+1} $ are on the same edge $ e_1^i \in S_i $, let $ u $ and $ v $ be the endpoints of $ e^i_1 $, where $ a_i $ is encountered before $ a_{i+1} $ when traversing $ e^i_1 $ from $ u $ to $ v $. If $ i=1$ then $ X_i = (a_1, v) $, otherwise $ X_i=(v) $, see Figure~\ref{fig:caseacrossingsquare}.
      \item If $ a_{i} $ and $ a_{i+1} $ are two corners of $ S_i $ that do not share an edge, let $ v $ be the corner of $ S_i $ to the right of $ \overrightarrow{a_ia_{i+1}} $. Then, $ X_i = (a_i, v, a_{i+1}) $, see Figure~\ref{fig:casebcrossingsquare}.
		    
		    \item If $ a_i $ and $ a_{i+1} $ belong to two adjacent edges $ e_1^i, e_2^i \in S_i $, let $ v $ be the corner of $ S_i $ shared by $ e_1^i $ and $ e_2^i $, and let $ u, u' $ be the other endpoints of $ e^i_1 $ and $ e^i_2 $, respectively. If $ a_i=u$ then $ X_i = (u, v, u') $, otherwise $ X_i = (v,u') $, see Figure~\ref{fig:caseccrossingsquare}.
		    
		    \item If $ a_i $ and $ a_{i+1} $ belong to the interior of two parallel edges $ e^i_1 $ and $ e^i_2 $, and the last point in $ X_{i-1} $ was the endpoint $v$ of $ e^i_1 $ to the left (resp., right) of~$ \overrightarrow{a_ia_{i+1}} $, $ X_i $ is the endpoint $v'$ of $ e^i_2 $ to the left (resp., right) of $ \overrightarrow{a_ia_{i+1}} $, see Figure~\ref{fig:casedcrossingsquare}.
		\end{enumerate}
	\end{definition}
	
	\begin{figure}[tb]
		\captionsetup[sub]{justification=centering}
		\centering
	        \begin{subfigure}[b]{0.18\textwidth}
				\centering
	    	    \includegraphics{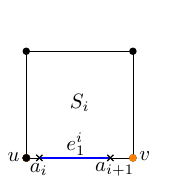}
	    	    \caption{}
	    	    \label{fig:caseacrossingsquare}
	    	\end{subfigure}
                \qquad
		    \begin{subfigure}[b]{0.18\textwidth}
   			    \centering
	    	    \includegraphics{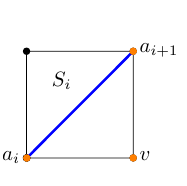}
	    	    \caption{}
	    	    \label{fig:casebcrossingsquare}
    	    \end{subfigure}
    	    \qquad
		    \begin{subfigure}[b]{0.18\textwidth}
   			    \centering
	    	    \includegraphics{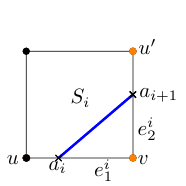}
	    	    \caption{}
	    	    \label{fig:caseccrossingsquare}
    	    \end{subfigure}
    	    \qquad
		    \begin{subfigure}[b]{0.18\textwidth}
   			    \centering
	    	    \includegraphics{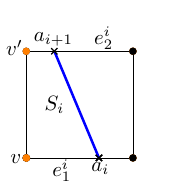}
	    	    \caption{}
	    	    \label{fig:casedcrossingsquare}
    	    \end{subfigure}
    	    \caption{Some of the positions of the intersection points between $ \mathit{SP_w}(s,t) $ (blue) and a square cell (notice that $ a_i $ and $ a_{i+1} $ could have other positions). The orange corners represent the vertices of the crossing path $ X(s,t) $.}
        \end{figure}

        \begin{lemma}
            The vertex sequence $ (X_1, \ldots, X_n) $ in Definition~\ref{def:crossingsquare} defines a grid path.
        \end{lemma}

        \begin{proof}
            Note that each sequence $ X_i, i\in \{1, \ldots, n\} $, is determined by an ordered sequence of adjacent vertices in $ G_{4\text{corner}} $.
            
            Now, there are two possible positions for $ a_{i+1} $: the interior of an edge, or a corner. Suppose, w.l.o.g, that $ a_{i+1} $ belongs to the interior of the edge $ e_1^{i+1} $ in cell $ S_{i+1} $ containing $ a_{i+1}$ and $ a_{i+2}$. Then there are four possible positions for $ a_i $ on the edges of $ S_i $, up to symmetries: an endpoint of $ e_1^{i+1} $, the interior of an edge adjacent to $ e_1^{i+1} $, the other endpoint of an edge adjacent to $ e_1^{i+1} $, or the interior of an edge parallel to $ e_1^{i+1} $. In each case, $ X_{i} $ either (i) contains one endpoint of the edge containing $ a_{i+1} $, and then $ X_{i+1} $ contains the other, or (ii) $ X_{i} $ contains the two endpoints, and then $ X_{i+1} $ contains the last vertex of $ X_{i} $ or does not contain any endpoint of $ e_1^{i+1} $. Otherwise, if $ a_{i+1} $ is a corner, there are three possible positions for $ a_i $, up to symmetries, see parts 1, 2 and 3 in Definition~\ref{def:crossingsquare}. In each case, the last vertex of $ X_{i} $ is $a_{i+1}$. In addition, either this vertex is not in $ X_{i+1}$, or it is the first vertex of $ X_{i+1}$. Thus, the sequence $ (X_1, \ldots, X_n) $ defines a path on the tessellation.
        \end{proof}
        
        In order to upper-bound the ratio $ \frac{\lVert \mathit{X}(s,t)\rVert}{\lVert \mathit{SP_w}(s,t)\rVert} $, we can use a variant of the mediant inequality, see Observation~\ref{thm:1}, and analyze the polygons resulting from consecutive intersection points between $\mathit{SP_w}(s,t)$ and $ X(s,t)$. Let $ (s\!=\!u_1, u_{2}, \ldots, u_\ell\!=\!t) $ be the sequence of consecutive points where $ X(s, t) $ and~$ \mathit{SP_w}(s, t) $ coincide, see Figure~\ref{fig:crossingsquare} for an illustration; in case $ \mathit{SP_w}(s, t) $ and~$X(s, t) $ share at least one segment, the corresponding points are defined as the endpoints of each of these segments.
        
        \begin{observation}
            \label{thm:1}	
		Let $\mathit{X}(s, t) $ be a crossing path defined from a weighted shortest path $ \mathit{SP_w}(s,t) $. Let $ u_i $ and $ u_{i+1} $ be two consecutive points where $ \mathit{SP_w}(s,t) $ and $\mathit{X}(s,t) $ coincide. Then, $ \frac{\lVert\mathit{X}(s,t)\rVert}{\lVert \mathit{SP_w}(s,t)\rVert} \leq \max_{1 \leq i \leq \ell-1}{\frac{\lVert\mathit{X}(u_i,u_{i+1})\rVert}{\lVert \mathit{SP_w}(u_i,u_{i+1})\rVert}} $.
	\end{observation}
        
        The union of $ \mathit{SP_w}(s,t) $ and $ X(s,t) $ between two consecutive points $ u_j $ and $ u_{j+1} $, for $ 1 \leq j < \ell $, induces a (simple) polygon, which degenerates into a path if and only if $ \mathit{SP_w}(u_j,u_{j+1}) $ and $ X(u_j,u_{j+1}) $ coincide. The boundary of these polygons consist of the portions of $\mathit{SP_w}(s,t)$ and $ X(s,t)$ between consecutive intersection points.
        These polygons are the basic component that we analyze to obtain our main result. Depending on the number of cells and edges intersected by $ X(u_j, u_{j+1}) $ and $ \mathit{SP_w}(u_j, u_{j+1})$ we distinguish two different types of polygons. The first type is denoted~$ P_k^1 $, for $ k \in \{0,1,2\} $, see Figure~\ref{fig:weakly1square}. The second type is denoted~$ P_k^\ell $, for $ k \in \{1,2\} $ and $ \ell\geq 2$, see Figure~\ref{fig:weakly3square}. Informally, $ \ell $ represents the number of cells whose interior is intersected by $ \mathit{SP_w}(u_j, u_{j+1})$, and $ k $ is the number of different edges of the last cell whose interior is intersected by $ \mathit{SP_w}(u_j, u_{j+1}) \cup X(u_j, u_{j+1}) $ minus $ 1 $.

    \begin{definition}
		\label{def:weakly1square}
		Let $ u_j, u_{j+1} \in S_i $ be two consecutive points where $ \mathit{SP_w}(s, t) $ and~$ X(s, t) $ coincide. A polygon induced by $ u_j $ and $ u_{j+1} $ is of \emph{type}~$ P^1_k$, for $ k\in\{0,1,2\} $, if $ \mathit{SP_w}(u_j, u_{j+1}) \cup X(u_j, u_{j+1}) $ intersects the interior of~$ k+1 $ different edges of~$ S_i $.
	\end{definition}
	
	\begin{figure}[tb]
		\centering
		\includegraphics{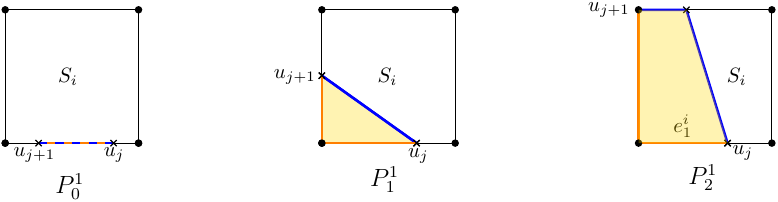}
		\captionof{figure}{Some polygons of type $ P^1_k $, and the subpaths $ \mathit{SP_w}(u_{j},u_{j+1}) $ (blue) and $ X(u_j,u_{j+1}) $ (orange) in a square mesh.}
		\label{fig:weakly1square}
	\end{figure}
	
	\begin{proposition}
	    \label{prop:onlyvertex}
	    Let $ u_j, u_{j+1} \in S_i $ be two consecutive points where $ \mathit{SP_{w}}(s,t) $ and $ X(s,t) $ coincide. $ P_0^1 $ and $ P_1^1$ are the only types of polygons that can arise in $ \mathit{SP_w}(u_j, u_{j+1}) \cup X(u_j, u_{j+1}) $ if the point $ u_j$ is a corner of a cell.
	\end{proposition}
	
	\begin{proof}
            First, note that, if $ \mathit{SP_w}(s,t) $ enters a cell $ S_i $ through one of its corners, by Definition~\ref{def:crossingsquare}, parts 1, 2 and 3, this corner, and the two endpoints of the edge containing the point where $ \mathit{SP_w}(s,t) $ leaves~$ S_i $ belong to~$ X(s,t) $. Hence,~$ \mathit{SP_w}(s,t) $ and $ X(s,t) $ intersect twice in~$ S_i $, once at~$ u_j $ and again where~$ \mathit{SP_w}(s,t) $ leaves~$ S_i $. This means that~$ u_{j+1} $ belongs to~$ S_i $.
            
            If the point $ u_{j+1} $ belongs to the same edge as $ u_j $, then $ \mathit{SP_w}(u_{j},u_{j+1}) $ and $ X(u_j, u_{j+1}) $ coincide, see Definition~\ref{def:crossingsquare}, part 1, and we have a polygon of type~$ P_0^1 $, by Definition~\ref{def:weakly1square}. If $ u_{j+1} $ belongs to an edge not containing $ u_j $ then, by Definition~\ref{def:crossingsquare}, parts 2 and 3, $ X(u_j,u_{j+1}) $ intersects the interior of $ 2 $ consecutive edges, and we have a polygon of type~$ P_1^1$.
	\end{proof}
	
	A key difference between triangular tessellations studied in~\cite{bose2023approximating} with the problem on square meshes is that in square meshes, two consecutive points $ u_j $ and $ u_{j+1}$ can belong to cells that are {\em not} adjacent. This can happen when $ \mathit{SP_w}(s,t) $ intersects the interior of two parallel edges on a cell, which is not possible in a triangular tessellation. Hence, we need to address this difficulty. We overcome this obstacle by defining more types of polygons. A property of these polygons is that the cells intersected by $ \mathit{SP_w}(u_j, u_{j+1}) $ are horizontally or vertically aligned. By Proposition~\ref{prop:onlyvertex}, we only need to consider the case where~$ u_j $ belongs to the interior of an edge.

	\begin{definition}
		\label{def:weakly3square}
		Let $ u_j $ and $ u_{j+1} $ be two consecutive points where $ \mathit{SP_w}(s, t) $ and~$ X(s, t) $ coincide. Let $ u_{j+1} $ be to the right (resp., left) of the line through~$ u_j $ perpendicular to the edge containing~$ u_j $, with respect to the $ SP_w(s,t) $, when considering $ SP_w(s,t) $ oriented from $ s $ to $ t $, see Figure~\ref{fig:weakly3square}. A polygon induced by $ u_j $ and $ u_{j+1} $ is of \emph{type}~$ P^{\ell}_k, \ \ell \geq 2, \ k \in \{1, 2\} $, if:
		\begin{itemize}
		\setlength\itemsep{0em}
		    \item $ \mathit{SP_w}(u_j, u_{j+1}) $ intersects the interior of $ \ell $ consecutive cells $ S_i, \ldots, S_m$, with $ \ell = m -i+1 $.
			\item $ X(u_j, u_{j+1}) $ contains all the vertices of the $ \ell $ cells intersected by $ \mathit{SP_w}(u_j, u_{j+1}) $ that are to the right (resp., left) of $ \overrightarrow{u_ju_{j+1}}$.
			\item $ \mathit{SP_w}(u_j, u_{j+1}) \cup X(u_j, u_{j+1}) $ intersects the interior of $k+1 $ different edges of~$ S_m$.
		\end{itemize}
	\end{definition}
		
	\begin{figure}[tb]
		\centering
		\includegraphics{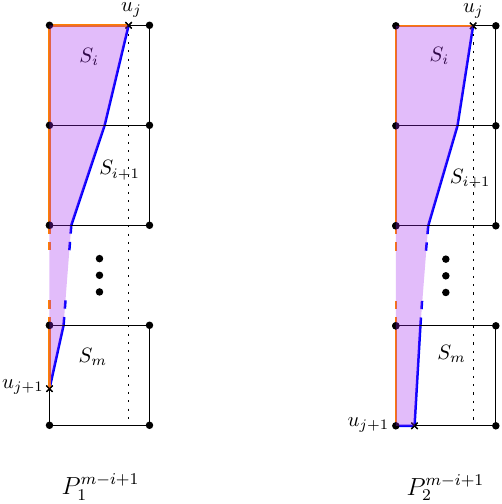}
		\captionof{figure}{Some polygons of type $ P^{m-i+1}_1 $ and $ P^{m-i+1}_2 $, and the subpaths $ \mathit{SP_w}(u_{j},u_{j+1}) $ (blue) and $ X(u_{j},u_{j+1}) $ (orange) in a square mesh.}
		\label{fig:weakly3square}
	\end{figure}
	
	\begin{proposition}
            \label{prop:unique}
	    Let $ u_j, u_{j+1} $ be two consecutive points where $ \mathit{SP_{w}}(s,t) $ and $ X(s,t) $ coincide. $ P_k^1 $, for $ k \in \{0,1,2\} $, and $ P_k^\ell $, for $ k \in \{1,2\} $ and $ \ell\geq 2$, are the only types of polygons that can arise in $ \mathit{SP_w}(u_j, u_{j+1}) \cup X(u_j, u_{j+1}) $.
	\end{proposition}
	
	\begin{proof}
	    Suppose $ u_j $ is a corner of $ \mathcal{S} $, then by Proposition~\ref{prop:onlyvertex} we have polygons of type~$ P_0^1$ and~$ P_1^1 $.
	    Now, suppose $ u_j$ is a point on the interior of an edge $ e $. Let~$ b $ be the next point where~$ \mathit{SP_w}(s, t) $ changes the cell(s) it belongs to. Then, if~$ b \in e $, by Definition~\ref{def:crossingsquare}, part 1, we attain a polygon of type~$ P_0^1$. If $ b $ belongs to an edge adjacent to $ e $, by Definition~\ref{def:crossingsquare}, part 3, we attain a polygon of type~$ P_1^1$. This follows from the same reasoning as in the proof of Proposition~\ref{prop:onlyvertex}. Otherwise, if~$ b $ belongs to the edge parallel to $ e $ in the same cell, there is a set of $\ell > 1$ cells, horizontally or vertically aligned with the cell containing $ u_j$ and $ b$, whose interior is intersected by~$ \mathit{SP_w}(u_j,u_{j+1})$. Note also that $ X(u_j,u_{j+1}) $ does not contain the four vertices of a cell in consecutive order. Hence, by Definition~\ref{def:crossingsquare}, part 4, we obtain a polygon of type~$ P_1^\ell $, if the last segment of $ \mathit{SP_w}(u_j,u_{j+1})$ properly intersects an edge of a cell, otherwise we obtain a polygon of type~$ P_2^{\ell}$.
	\end{proof}
	
	\subsubsection{Bounding the ratio for one polygon}
	The aforementioned polygons are an important tool in our proof since they are the only possible ways that $ \mathit{SP_w}(u_j,u_{j+1}) $ and $ X(u_j,u_{j+1}) $ interact. Therefore, it is enough to upper-bound $ \frac{\lVert X(u_j,u_{j+1})\rVert}{\lVert \mathit{SP_w}(u_j,u_{j+1})\rVert} $ when~$ X(u_j,u_{j+1})$ and $\mathit{SP_w}(u_j,u_{j+1})$ bound polygons of type~$ P_k^1$, $ k\in\{0,1,2\}$, and $ P_k^{\ell} $, $ k\in\{1,2\}, \ell \geq 2$.

        \begin{observation}
	    \label{obs:reduce0}
	    Let $ u_j, u_{j+1} $ be two consecutive points where $ \mathit{SP_{w}}(s,t) $ and $ X(s,t) $ coincide. In a polygon of type~$ P_0^1 $, the subpaths $ X(u_j,u_{j+1}) $ and $ \mathit{SP_w}(u_j, u_{j+1}) $ coincide. Hence, $ \frac{\lVert X(u_j, u_{j+1})\rVert}{\lVert \mathit{SP_w}(u_j, u_{j+1}) \rVert} = 1 $.
	\end{observation}

        Observation~\ref{obs:reduce0} implies that polygons of type~$ P_0^1 $ do not need to be considered when determining worst-case upper and lower bounds. Let $ p $ and $ q $ be two points on different edges of the same cell. Observation~\ref{obs:lengthsquares} gives the distance from $ p $ to $ q $. The result can be shown using the Pythagorean Theorem. 
	
	\begin{observation}
		\label{obs:lengthsquares}
		Let $ S_i $ be a square cell of side length $ 1 $, and let $ (v_1, v_2, v_3, v_4) $ be the four consecutive vertices of $ S_i $, in clockwise order. Let $ p \in [v_1, v_2] $ and $ q $ be a point on the boundary of $ S_i $. Then,
		\begin{enumerate}
			\setlength\itemsep{0em}
			\item If $ q \in [v_3, v_4] $, $ |pq| = \sqrt{a^2+b^2-2ab+1} $, where $ a=|pv_2|, b=|v_3q| $, see Figure~\ref{fig:46}.
			\item If $ q \in [v_2, v_3] $, $ |pq|=\sqrt{a^2+b^2} $, where $ a=|pv_2|, b=|v_2q| $, see Figure~\ref{fig:47}.
		\end{enumerate}
	\end{observation}
	
	\begin{figure}[tb]
		\centering
   		\begin{subfigure}[b]{0.44\textwidth}
   			\centering
	    	\includegraphics{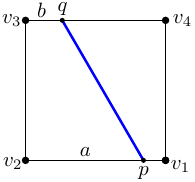}
	    	\caption{ Point $ q $ belongs to edge $ [v_3,v_4] $.}
	    	\label{fig:46}
    	\end{subfigure}
    	\qquad\quad
   		\begin{subfigure}[b]{0.44\textwidth}
   			\centering
	    	\includegraphics{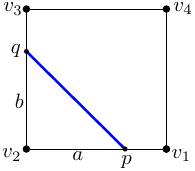}
	    	\caption{Point $ q $ belongs to edge $ [v_2,v_3] $.}
	    	\label{fig:47}
    	\end{subfigure}
    	\caption{The blue segment represents the subpath of $ \mathit{SP_w}(s,t) $ between two points $ p $ and $ q $ in a square cell.}
	\end{figure}
 
        Moreover, Observation~\ref{obs:reduce3} implies that a worst-case upper bound on the ratio $ \frac{\lVert X(u_j, u_{j+1})\rVert}{\lVert \mathit{SP_w}(u_j,u_{j+1}) \rVert} $ is {\em not} obtained when the paths $ \mathit{SP_w}(u_j,u_{j+1}) $ and $ X(u_j, u_{j+1}) $ bound polygons of type~$ P_2^\ell, \ \ell \geq 1 $.
	
	\begin{observation}
	    \label{obs:reduce3}
	    Let $ u_j, u_{j+1} $ be two consecutive points where $ \mathit{SP_{w}}(s,t) $ and $ X(s,t) $ coincide. Let $ u' $ be the last corner of a cell in $ \mathcal{S} $ intersected by $ X(u_j, u_{j+1}) $. A polygon of type~$ P_2^{\ell} $, for $ \ell \geq 1$, can be seen as a particular case of a polygon of type~$ P_1^{\ell+1} $, where the distance $ |u'u_{j+1}| $ tends to $ 0 $, see Figure~\ref{fig:weakly3square}. Thus, the ratio~$ \frac{\lVert X(u_j, u_{j+1})\rVert}{\lVert \mathit{SP_w}(u_j, u_{j+1}) \rVert} $ when $ \mathit{SP_w}(u_j,u_{j+1}) $ and $ X(u_j, u_{j+1}) $ bound a polygon of type~$ P_2^{\ell} $ is at most as large as when intersecting a polygon of type~$ P_1^{\ell+1} $.
	\end{observation}

        Before obtaining an upper bound on the ratio $ \frac{\lVert \mathit{X}(s, t)\rVert}{\lVert \mathit{SP_w}(s, t) \rVert} $ when $ X(s,t) $ and $ \mathit{SP_w}(s,t) $ bound a polygon of type~$ P_1^1$ or~$ P_1^\ell$, we first introduce the notion of $ P_1^1 $- and $ P_1^\ell$-triple of cells, see Figures~\ref{fig:newconfigsquare3} and \ref{fig:newconfigsquare4}, respectively.
	
	\begin{definition}
	    \label{def:Ptriples}
	    A \emph{$ P_1^\ell$-triple}, for $ \ell \geq 1 $, from a vertex $ s $ to a vertex $ t $ is defined as a set of $ \ell+4$ consecutive cells $ S_1, \ldots, S_{\ell+4}$ intersected by $ \mathit{SP_w}(s, t) $ with the following properties:
	    \begin{itemize}
                \setlength\itemsep{0em}
	        \item $ s $ is the vertex common to $ S_1$ and $ S_2$, and not adjacent to $ S_3$.
	        \item $ t $ is the vertex common to $ S_{\ell+3}$ and $ S_{\ell+4}$, and not adjacent to $ S_{\ell+2}$.
	        \item The union of $ \mathit{SP_w}(s, t) $ and $ X(s,t) $ determines two polygons of type~$ P_1^1$, and one polygon of type~$ P_1^\ell$ in between. 
	    \end{itemize}
	\end{definition}
	
	$ P_1^1$- and $ P_1^\ell$-triples are useful when $ \mathit{SP_w}(s, t) $ bounds several instances of a polygon of type $ P_1^1$ or $ P_1^\ell$, respectively, because they allow us to lessen the amount of cases we need to analyze. In the following, given that the ratio $ \frac{\lVert \mathit{X}(s, t)\rVert}{\lVert \mathit{SP_w}(s, t) \rVert} $ is maximized when $ \mathit{SP_w}(s, t) $ and $ X(s,t) $ bound a polygon $ P' $ of type~$ P_1^\ell, \ell \geq 1$, we create a specific instance of a $ P_1^\ell$-triple with the same properties as the former instance containing $ P'$, and upper-bound the ratio $ \frac{\lVert \mathit{SGP_w}(s, t)\rVert}{\lVert \mathit{SP_w}(s, t) \rVert} $ on that triple. The reason for creating the $ P_1^\ell$-triples is that they allow us to perform the optimizations over a simplified problem.
	
	\begin{figure}[tb]
        \centering
		\begin{subfigure}[t]{0.42\textwidth}
   		    \centering
	    	\includegraphics{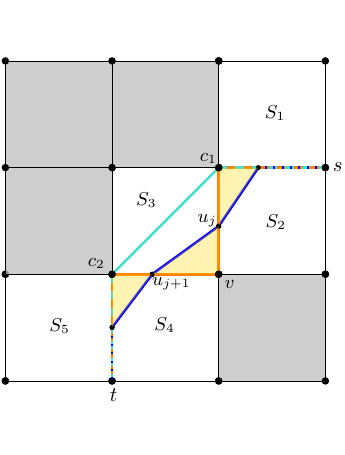}
	    	\caption{$ P_1^1 $-triple from $ s $ to $ t $ is represented in white. Shortcut path $ \Pi_3^1(s,t) $ is depicted in cyan. Observe that $ \mathit{SP_{w}}(s,t)$, $ X(s,t) $ and $ \Pi_3^1(s,t) $ coincide when~$ \mathit{SP_{w}}(s,t)$ coincides with the edges of the cells.}
	    	\label{fig:newconfigsquare3}
    	\end{subfigure}
	    \qquad\quad
		\begin{subfigure}[t]{0.42\textwidth}
		    \centering
	    	\includegraphics{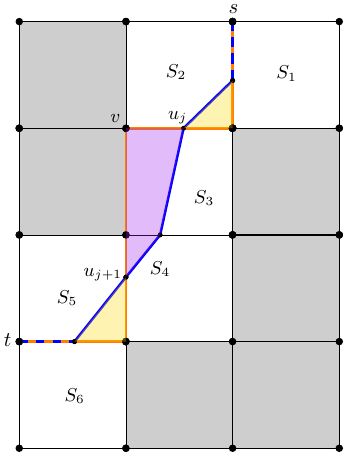}
	    	\caption{White cells represent the $ P_1^2 $-triple from $ s $ to $ t $.}
	    	\label{fig:newconfigsquare4}
        \end{subfigure}
    	\caption{$ \mathit{SP_w}(s, t) $, and $ X(s,t) $ are depicted in blue and orange, respectively.}
	\end{figure}
	
	\begin{lemma}
	    \label{lem:Ptriple}
	    Let $ v $ be the vertex common to $ S_2 $ and $ S_3 $, and not adjacent to $ S_1$. For any polygon $ P'$ of type~$ P_1^\ell, \ \ell \geq 1 $, a $ P_1^\ell$-triple can be defined such that $ v $ is the corner where the crossing path makes a right angle turn in $P'$, and the weights $ \omega_2, \ldots, \omega_{\ell+3} $ of the cells in the $ P_1^\ell$-triple remain the same as in $ P' $. We say that this $ P_1^\ell$-triple \emph{corresponds} to the polygon $ P'$.
	\end{lemma}
	
	\begin{proof}
	    Let $ u_j, u_{j+1}$ be the points where $ \mathit{SP_w}(s, t) $ respectively enters and leaves~$ P'$, and let $ u_j \in S_{j'-1} \cap S_{j'} $ and $ u_{j+1} \in S_{j'+\ell-1} \cap S_{j'+\ell} $, for some $ j' \leq n-\ell$. Consider a $ P_1^\ell$-triple intersected by a path $ \mathit{SP_w}(s', t')$ where the weights of the cells~$ S_2, \ldots, S_{\ell+3} $ are, respectively, the same as the weights of cells~$ S_{j'-1}, \ldots, S_{j'+\ell}$ in the former instance. The weights of the cells $ S_1$ and $ S_{\ell+4}$ are obtained by solving the system of equations given by Snell's law of refraction. By construction, this $ P_1^\ell$-triple is a valid instance of $ \mathit{SP_w}(s', t')$, and it intersects the cells~$ S_{3}, \ldots, S_{\ell+2}$ forming the same polygon $ P'$ as in the former instance.
	\end{proof}
		
	Lemma~\ref{lem:Ptriple} allows us to assume, from now on, that the maximum value the ratio $ \frac{\lVert \mathit{SGP_w}(s,t) \rVert}{\lVert \mathit{SP_w}(s,t) \rVert} $ can take is given when $ \mathit{SP_w}(s,t)$ intersects a $ P_1^1$-triple or a $ P_1^\ell$-triple.
	
	Now, using the previous observations and the following lemma, we prove that we only need to analyze polygons of type~$ P^1_1 $. Recall that $ P_1^1$ is the type of polygon where $ \mathit{SP_{w}}(u_j,u_{j+1}) \cup X(u_j,u_{j+1}) $ intersects the interior of 2 different edges of the cell containing $ u_j $ and $ u_{j+1} $, see Figure~\ref{fig:weakly1square}.
    
		\begin{lemma}
			\label{lem:16}
			 The ratio $ \frac{\lVert \mathit{SGP_{w}}(s,t)\rVert}{\lVert \mathit{SP_{w}}(s,t)\rVert} $ in a $ P^\ell_1 $-triple is upper-bounded by the ratio~$\frac{\lVert X(u_j, u_{j+1})\rVert}{\lVert \mathit{SP_w}(u_j, u_{j+1}) \rVert}$ when $ X(u_j,u_{j+1}) $ and $ \mathit{SP_{w}}(u_j,u_{j+1}) $ bound a polygon of type~$ P^1_1 $, for any pair $ (u_j, u_{j+1}) $ of consecutive points where $ \mathit{SP_{w}}(s,t) $ and $ X(s,t) $ coincide.
		\end{lemma}
		
		\begin{proof}
		Let $ (S_1, \ldots, S_{\ell+4}) $ be the ordered sequence of $ \ell+4 $ consecutive square cells intersected by $ \mathit{SP_w}(s,t) $ in a~$ P_1^\ell $-triple. Let $ (v^k_1, v^k_2, v^k_3, v^k_4), 2 \leq k \leq \ell+2$, be the four corners of cell $ S_k $, in clockwise order, where $ v_1^3 $ is the vertex common to $ S_2 $ and $ S_3 $, and not adjacent to $ S_1$, and $ v_m^{k} $ is on the same horizontal/vertical line as~$ v_m^{k'} $, for $ k' > k $ and $ m \in \{1,2,3,4\}$, see Figure~\ref{fig:48}.

        \begin{figure}[tb]
			\centering
			\includegraphics[width=\textwidth]{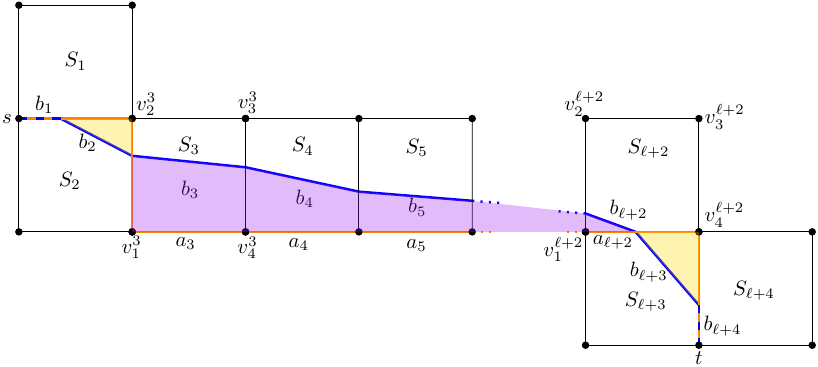}
			\captionof{figure}{$ P^{\ell}_1 $-triple. $ \mathit{SP_w}(s,t) $ and $ X(s,t) $ are depicted in blue and orange, respectively.}
			\label{fig:48}
		\end{figure}

        We now define a grid path that allows us to upper bound the worst-case ratio $ \frac{\lVert \mathit{SGP_{w}}(s,t)\rVert}{\lVert \mathit{SP_{w}}(s,t)\rVert} $ in a $ P^\ell_1 $-triple. This path intersects the edge adjacent to the cell with minimum weight. Thus, we split the polygon of type $ P^\ell_1 $ in two other polygons to demonstrate that the worst-case upper bound is obtained when $ \mathit{SGP_w}(s, t) $ intersects the interior of one cell.
			
		Let $ S_{i} $ be the cell with minimum weight among the cells $ S_3, \ldots, S_{\ell+2} $. Consider the grid path $ \mathit{\Pi}(s,t) = (s=v_2^2,v_2^3, \ldots, v_2^{i}, v_1^{i}) \cup X(v_4^{i}, t) $, see violet path in Figure~\ref{fig:splitpath}. We know that $ A = \frac{\lVert \mathit{SGP_w}(s, t)\rVert}{\lVert \mathit{SP_w}(s, t) \rVert} \leq \frac{\lVert \mathit{\Pi}(s, t)\rVert}{\lVert \mathit{SP_w}(s, t) \rVert} $. In addition, the union of $ \mathit{SP_w}(s, t) $ and $ \mathit{\Pi}(s, t)$ induces two polygons of type~$ P_0^1 $, one or two (if $i=3$ or $i=\ell+2$) polygons of type~$ P_1^1 $, one polygon of type~$ P_1^{i-2} $, and one polygon of type~$ P_1^{\ell-i+3} $.

        \begin{figure}[tb]
			\centering
			\includegraphics[width=\textwidth]{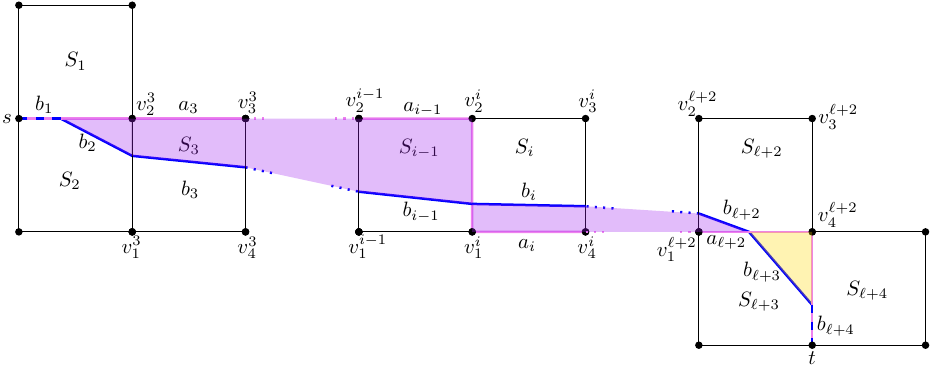}
			\captionof{figure}{Path $ \mathit{\Pi}(s,t) $ depicted in violet. Note the polygons of type $ P_1^{i-2}$ and $ P_1^{\ell-i+3}$, respectively to the left and to the right of the edge $ (v_1^i, v_2^i) $.}
			\label{fig:splitpath}
		\end{figure}
  
		By the mediant inequality, see Observation~\ref{thm:1}, we just need to consider the maximum among the ratios $ \frac{\lVert \mathit{\Pi}(u_j, u_{j+1})\rVert}{\lVert \mathit{SP_w}(u_j, u_{j+1}) \rVert} $, for $ u_j, u_{j+1} $ being two consecutive points where~$ \mathit{SP_w}(s, t) $ and~$ \mathit{\Pi}(s, t)$ coincide. By Observation~\ref{obs:reduce0}, the ratio is not maximized when~$ \mathit{SP_w}(s, t) $ bounds the polygons of type~$ P_0^1$. It will turn out that $P_1^1$ is the only relevant type of polygon. So, let us upper-bound the ratio when $ \mathit{SP_w}(s, t) $ bounds the polygon of type~$ P_1^{\ell-i+3}$. Note that the analysis is similar when $ \mathit{SP_w}(s, t) $ bounds the polygon of type~$ P_1^{i-2}$.
  
        Let $ u_j \in (v^i_1, v^i_2) $, and $ u_{j+1} \in (v^{\ell+2}_1, v^{\ell+2}_4] $ in the polygon of type~$ P^{\ell-i+3}_1 $. The weighted length of each segment of $ \mathit{\Pi}(u_j,u_{j+1}) $ bounding cell $ S_{\ell'}, \ i \leq \ell' \leq \ell+1 $, is $\min\{\omega_{\ell'}, \omega_{\ell''}\} \leq \omega_{\ell'} $, where $ \omega_{\ell''} $ is the weight of the cell sharing the edge of $ \mathit{\Pi}(u_j,u_{j+1}) $ with $ S_{\ell'} $. Thus,
		\[
			A \leq \frac{a\omega_i+\omega_i+\omega_{i+1}+\ldots+\omega_{\ell+1}+a_{\ell+2}\min\{\omega_{\ell+2}, \omega_{\ell+3}\}}{b_i\omega_i+b_{i+1}\omega_{i+1}+\ldots+b_{\ell+2}\omega_{\ell+2}},
		\]		
		where $ a = |u_jv^i_1| $, $ a_{\ell+2} = |v^{\ell+2}_1u_{j+1}| $, and $ b_{\ell'}, \ i \leq \ell' \leq \ell+2 $, is the length of the subpath of $ \mathit{SP_w}(u_j, u_{j+1}) $ traversing cell $ S_{\ell'} $. We can also write the ratio $ A $ as:

       \begin{equation*}
		    A \leq \frac{\sum_{\ell'=i}^{\ell+2}(a_{\ell'}\omega_{\ell'}+c_{\ell'}\omega_i)}{\sum_{\ell'=i}^{\ell+2}b_{\ell'}\omega_{\ell'}},
		\end{equation*}
  
        where $ a_{\ell'} $ and $ c_{\ell'} $ are, respectively, the lengths of the projection of $ b_{\ell'} $ on~$ [v_1^{\ell'}, v_4^{\ell'}] $ and $ [v_1^{\ell'}, v_2^{\ell'}] $. By the mediant inequality, we have that
			
		\begin{equation*}
		    A \leq \max_{i\leq\ell'\leq\ell+2}\frac{a_{\ell'}\omega_{\ell'}+c_{\ell'}\omega_i}{b_{\ell'}\omega_{\ell'}} \leq \max_{i\leq\ell'\leq\ell+2}\frac{a_{\ell'}\omega_{\ell'}+c_{\ell'}\omega_{\ell'}}{b_{\ell'}\omega_{\ell'}} = \max_{i\leq\ell'\leq\ell+2}\frac{a_{\ell'}+c_{\ell'}}{b_{\ell'}},
		\end{equation*}
		where the last inequality is obtained by taking into account that $ \omega_i $ is the minimum weight of the cells $ S_3, \ldots, S_{\ell+2} $. Then, the maximum value $ A $ is the ratio $ \frac{\lVert X(u_{j'}, u_{j'+1})\rVert}{\lVert \mathit{SP_w}(u_{j'}, u_{j'+1}) \rVert} $ when $ \mathit{SP_w}(u_{j'},u_{j'+1}) $ and $ X(u_{j'},u_{j'+1}) $ bound a polygon of type~$ P^{1}_1 $.
	\end{proof}
		
		Hence, according to Observations~\ref{obs:reduce0} and \ref{obs:reduce3} and Lemma~\ref{lem:16}, we upper-bound the ratio $ \frac{\lVert X(u_j, u_{j+1})\rVert}{\lVert \mathit{SP_w}(u_j, u_{j+1}) \rVert} $ for $X(u_j, u_{j+1})$ and $ \mathit{SP_w}(u_j, u_{j+1})$ by bounding a polygon of type $ P^1_1 $. However, it might happen that the ratio $ \frac{\lVert \mathit{X}(s, t)\rVert}{\lVert \mathit{SP_w}(s, t) \rVert} $ is not a tight upper bound for the ratio $ \frac{\lVert \mathit{SGP_w}(s, t)\rVert}{\lVert \mathit{SP_w}(s, t) \rVert} $. See, for instance, Figure~\ref{fig:badratiosquare2}, where $ \frac{\lVert \mathit{X}(s, t)\rVert}{\lVert \mathit{SP_w}(s, t) \rVert} = \sqrt{2} $, but $ \frac{\lVert \mathit{SGP_w}(s, t)\rVert}{\lVert \mathit{SP_w}(s, t) \rVert} =1 $. Note that $ \mathit{SGP_w}(s, t) $ is the straight-line segment between $ s $ and $ t $.
		
		\begin{figure}[tb]
			\centering
			\includegraphics[scale=1]{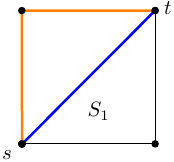}
			\captionof{figure}{The ratio between the weight of the crossing path $ X(s,t) $, depicted in orange, and the weight of $ \mathit{SP_w}(s,t) $, in blue, is $ \sqrt{2} $. Recall that the weight of the cells adjacent to $ S_1 $, which are not depicted, is infinity.}
			\label{fig:badratiosquare2}
		\end{figure}
		
		In order to overcome this difficulty we define a class of shortcut paths for square mesh.

		\begin{definition}
			\label{def:shortcut2square}
			Let $ (v^i_1, v^i_2, v^i_3, v^i_4) $ be the sequence of vertices of a cell $ S_i $, in clockwise or counterclockwise order. Let $ \mathit{SP_w}(s, t) $ enter cell $ S_i $ through the edge~$ [v_j^i, v_{\text{mod (j,4)}+1}^i] $, for some $ j\in\{1,\ldots,4\}$. If $ X(s,t) $ contains the subpath $ (v_j^i, v^i_{\text{mod (j,4)}+1}, v_{\text{mod(j+1, 4)}+1}^i) $, the \emph{shortcut path} $ \Pi_i^1(s,t) $ is defined as the grid path $ X(s,v^i_j)\cup (v^i_j, v^i_{\text{mod(j+1,4)}+1}) \cup X(v^i_{\text{mod(j+1, 4)}+1}, t) $, see cyan path in Figure~\ref{fig:shortcut1square}. Symmetrically, if $ X(s,t) $ contains the subpath $ (v_{\text{mod(j, 4)}+1}^i, v^i_j, v_{\text{mod(j+2, 4)}+1}^i) $, $ \Pi_i^1(s,t) $ is defined as the grid path $ X(s,v^i_{\text{mod(j, 4)}+1})\cup (v^i_{\text{mod(j, 4)}+1}, v^i_{\text{mod(j+2, 4)}+1}) \cup X(v^i_{\text{mod(j+2, 4)}+1}, t) $.
		\end{definition}

    \begin{figure}[tb]
		\centering
		\includegraphics{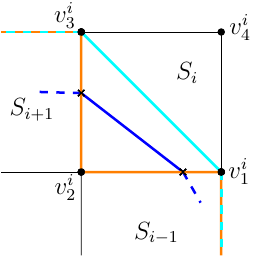}
	\caption{Polygon of type~$ P^1_1 $ intersecting cell~$ S_i $. Subpaths of the shortcut path $ \Pi_i^1(s,t) $ (cyan), the crossing path $ X(s,t) $ (orange), and $ \mathit{SP_w}(s,t) $ (blue).}
    	\label{fig:shortcut1square}
	\end{figure}		
		
	 So, now, we use the shortcut path specified in Definition~\ref{def:shortcut2square} to relate the weights of the square cells intersected by $\mathit{SP_w}(u_j,u_{j+1})$ in a $ P_1^1$-triple. 

		\begin{lemma}
			\label{lem:equallength2square}
			Let $ \frac{\lVert \mathit{SGP_w}(s',t') \rVert}{\lVert \mathit{SP_w}(s',t') \rVert} $ be the maximum ratio attained when $ \mathit{SP_w}(s',t') $ bounds a polygon of type~$ P_1^1 $. Consider the corresponding $ P_1^1$-triple between~$ s $ and $ t $. Let $ S_3 $ be the cell where the maximum ratio $ \frac{\lVert \mathit{X}(u_j,u_{j+1}) \rVert}{\lVert \mathit{SP_w}(u_j,u_{j+1}) \rVert} $ is attained, where $ u_j, \ u_{j+1} \in S_3 $ are two consecutive points where $ \mathit{SP_{w}}(s,t) $ and $ X(s,t) $ coincide. Then, $ \lVert X(s,t) \rVert = \lVert \Pi_3^1(s,t) \rVert $.
		\end{lemma}
		
		\begin{proof}
			We prove the result by contradiction, arguing that if there is at least one grid path $\mathit{GP_w}(s, t)$ among $ \{X(s, t), \Pi_3^1(s,t)\} $, whose weighted length is strictly less than the weighted length of the other grid path, then this instance cannot maximize $ \frac{\lVert \mathit{SGP_w}(s',t')\rVert}{\lVert \mathit{SP_w}(s',t')\rVert} $. We will show this by finding another assignment of weights $ w' $ for the cells $ S_1, \ldots, S_5 $, such that $ \frac{\lVert\mathit{GP_{w'}}(s,t)\rVert}{\lVert \mathit{SP_{w'}}(s,t)\rVert} > \frac{\lVert\mathit{GP_w}(s,t)\rVert}{\lVert \mathit{SP_w}(s,t)\rVert} $, proving that the given instance does not provide the maximum ratio.
		    
		    We first set to infinity the weight of all the cells that are not traversed by~$ X(s,t) $. This way, we ensure that when modifying the weights of some cells, the combinatorial structure of the shortest path is preserved. Let $ u_j, \ u_{j+1} $ be two consecutive points where $ X(s,t) $ and $ \mathit{SP_w}(s,t) $ coincide, see Figure~\ref{fig:newconfigsquare3}. The weight of the crossing path $ X(s,t) $ along the edges of $ S_3 $ is $ \min\{\omega_{2}, \omega_3\} + \min\{\omega_3, \omega_{4}\} $, and the weight of the shortcut path $ \Pi_3^1(s,t) $ through $ S_3 $ is $ \sqrt{2}\omega_{3} $. Recall that the edges of the cells have length $ 1 $, hence the length between two non-adjacent vertices of the same cell is $ \sqrt{2}$. Let $ c_1 $ be the common corner between $ S_{1} $ and $ S_{3} $, let $ c_2 $ be the common corner between $ S_{3} $ and $ S_{5} $.
		    
			\begin{itemize}
				\setlength\itemsep{0em}
				 \item If $ \mathit{GP_w}(s,t) = \Pi_3^1(s,t) $, then $ \lVert \Pi_3^1(s,t) \rVert < \lVert X(s,t) \rVert $, and we have that
				 \begin{equation*}
		            \frac{\lVert \mathit{GP_w}(s,t) \rVert}{\lVert \mathit{SP_w}(s,t) \rVert} = \frac{\lVert \Pi_3^1(s,c_1) \rVert + \sqrt{2}\omega_{3} + \lVert \Pi_3^1(c_2,t) \rVert}{\lVert \mathit{SP_w}(s,u_{j}) \rVert +  |u_{j}u_{j+1}|\omega_3 + \lVert \mathit{SP_w}(u_{j+1},t) \rVert}.
		        \end{equation*}
		        
				\item If $\mathit{GP_w}(s,t) = X(s,t) $, then $ \lVert X(s,t) \rVert < \lVert \Pi_3^1(s,t) \rVert $. Thus, we have
		        \begin{equation*}
		            \frac{\lVert \mathit{GP_w}(s,t) \rVert}{\lVert \mathit{SP_w}(s,t) \rVert} = \frac{\lVert X(s,c_1) \rVert + \min\{\omega_{2}, \omega_3\} + \min\{\omega_{3}, \omega_4\} + \lVert X(c_2,t) \rVert}{\lVert \mathit{SP_w}(s,u_{j}) \rVert + |u_ju_{j+1}|\omega_3 + \lVert \mathit{SP_w}(u_{j+1},t) \rVert}.
		        \end{equation*}
		        
	        \end{itemize}
		        
		    The ratio $ \frac{\lVert\mathit{GP_w}(s,t)\rVert}{\lVert \mathit{SP_w}(s,t)\rVert} $ is a strictly monotonic function for every $ \omega_k $~\cite{chew1984pseudolinearity, RAPCSAK1991353}. So, if this function is decreasing in the direction of~$ \omega_3 $, we can decrease the weight~$ \omega_3 $. Otherwise, we can increase the weight~$ \omega_3 $. In both cases, we can increase the ratio.
		        
		    Thus, we found another weight assignment $ w' $ such that $ \frac{\lVert\mathit{GP_{w'}}(s,t)\rVert}{\lVert \mathit{SP_{w'}}(s,t)\rVert} > \frac{\lVert\mathit{GP_w}(s,t)\rVert}{\lVert \mathit{SP_w}(s,t)\rVert} $. In addition, the change in $ \omega_3$ can be as small as needed so that the weight of~$ \mathit{GP_{w'}}(s,t) $ is not larger than that of the other grid path in $ \{X(s, t), \Pi_3^1(s,t)\} $ with the new weight assignment $ w'$, and hence, the shortest path among the two grid paths in the set does not change.
		\end{proof}
		
		Lemma~\ref{lem:equallength2square} implies that we just need to upper-bound the ratio~$ \frac{\lVert X(u_j, u_{j+1})\rVert}{\lVert \mathit{SP_w}(u_j, u_{j+1}) \rVert} $ when the shortest path intersects a~$ P_1^1$-triple and the weight of the grid paths intersecting this triple is the same. Lemma~\ref{lem:19} presents an upper bound on this ratio. Since the exact shape of $\mathit{SP_w}(s,t)$ is unknown, when computing the upper bound in Lemma~\ref{lem:19}, we will maximize the upper bound for any possible position of the points $ u_j$ and $ u_{j+1}$ on the edges of $ \mathcal{S} $.
		
		\begin{lemma}
			\label{lem:19}
			Let $ u_j, \ u_{j+1} \in S_3 $ be two consecutive points where a shortest path~$ \mathit{SP_w}(s,t) $ and the crossing path $ X(s,t) $ coincide in a $ P_1^1$-triple. If $ u_j $ and $ u_{j+1} $ induce a polygon of type~$ P_1^1$, and $ \lVert X(s,t)\rVert = \lVert \Pi_3^1(s,t)\rVert $, then $ \frac{\lVert X(u_j, u_{j+1})\rVert}{\lVert \mathit{SP_w}(u_j, u_{j+1}) \rVert} \leq \frac{2}{\sqrt{2+\sqrt{2}}} \approx 1.08 $.
		\end{lemma}
		
		\begin{proof}
			 Let $ (v^3_1, v^3_2, v^3_3, v^3_4) $ be the sequence of vertices on the boundary of cell $ S_3 $ in clockwise order. Since $ u_j$ and $ u_{j+1}$ are two consecutive points where $ \mathit{SP_w}(s,t) $ and $ X(s,t) $ coincide, and they induce a polygon of type~$ P_1^1$, $ \mathit{SP_w}(s,t) $ enters $ S_3 $ from cell $ S_{2} $ through $u_j$, and $ \mathit{SP_w}(s,t) $ leaves $ S_3 $ and enters cell $ S_{4} $ through $ u_{j+1}$. Suppose, without loss of generality, that $ u_j \in [v^3_1, v^3_2) $, and $ u_{j+1} \in (v^3_2, v^3_3] $, see Figure~\ref{fig:shortcutP2square}. Let $ a, b, c $ be the lengths $ |u_jv^3_2|, |v^3_2u_{j+1}| $, and $ |u_ju_{j+1}| $, respectively. According to Observation~\ref{obs:lengthsquares}, part $ 2 $, $ c = \sqrt{a^2+b^2} $. We want to maximize the ratio $ \frac{\lVert X(u_j, u_{j+1})\rVert}{\lVert \mathit{SP_w}(u_j, u_{j+1}) \rVert} $ for all possible weight assignments of $ \omega_2, \omega_3$, and $ \omega_4$.

            \begin{figure}[tb]
		\centering
		\includegraphics{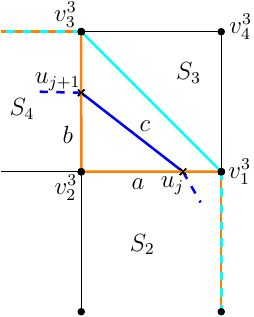}
	    	\caption{Polygon of type~$ P^1_1 $ intersecting cell $ S_3 $. Subpaths of the shortcut path $ \Pi_3^1(s,t) $ (cyan), the crossing path $ X(s,t) $ (orange), and $ \mathit{SP_w}(s,t) $ (blue).}
	    	\label{fig:shortcutP2square}
	\end{figure}
			 
			 Traversing cell $ S_3 $ there is also the grid path $ \Pi_3^1(s,t) $. Since $ \lVert X(s,t)\rVert = \lVert \Pi_3^1(s,t)\rVert \Rightarrow \min\{\omega_{2}, \omega_3\}+\min\{\omega_3, \omega_{4}\} = \sqrt{2}\omega_3 $. Thus,
			\begin{equation*}
				\begin{cases}
			 		\omega_{2}+\omega_{4} = \sqrt{2}\omega_3 & \mbox{if } \omega_{2}, \omega_{4} \leq \omega_3 \\
			 		\omega_{2} = (\sqrt{2}-1)\omega_3 & \mbox{if } \omega_{2} < \omega_3 < \omega_{4} \\
			 		\omega_{4} = (\sqrt{2}-1)\omega_3 & \mbox{if } \omega_{4} < \omega_3 < \omega_{2} \\
			 	\end{cases}.
			\end{equation*}
   
			Note that we do not take into account the case $ \omega_3 \leq \omega_{2}, \omega_{4} $ because it would lead to $ 4w_3=\sqrt{8}w_3 \Rightarrow w_3=0 $. If this is the case, the subpaths $ X(u_j,u_{j+1}) $ and $\mathit{SP_w}(u_j, u_{j+1})$ have the same weighted length, so the ratio is $1$.
   
   Then, the ratio $ R= \frac{\lVert X(u_j, u_{j+1})\rVert}{\lVert \mathit{SP_w}(u_j, u_{j+1})\rVert} $ when $ \mathit{SP_w}(u_j, u_{j+1}) $ and $ X(u_j, u_{j+1}) $ bound a polygon of type~$ P^1_1 $ is upper-bounded by:
			
			\begin{align*}
				R & = \frac{a\min\{\omega_{2}, \omega_{3}\}+b\min\{\omega_3, \omega_{4}\}}{c\omega_3}
				\begin{cases}
					\overbrace{=}^{\text{if } \omega_{2}, \omega_{4} \leq \omega_3} \frac{a\omega_{2}+b\omega_{4}}{c\omega_3}
					\begin{cases}
						\overbrace{\leq}^{\text{if } a \leq b} \frac{a(\sqrt{2}\omega_3-\omega_{4})+b\omega_{4}}{c\omega_3} = \\
						\overbrace{\leq}^{\text{if } b < a} \frac{a\omega_{2}+b(\sqrt{2}\omega_3-\omega_{2})}{c\omega_3} =
					\end{cases} \\
					\overbrace{=}^{\text{if } \omega_{2} < \omega_3 < \omega_{4}} \frac{a\omega_{2}+b\omega_3}{c\omega_3} \leq \frac{a(\sqrt{2}-1)\omega_3+b\omega_3}{c\omega_3} = \\
					\overbrace{=}^{\text{if } \omega_{4} < \omega_3 < \omega_{2}} \frac{a\omega_3+b\omega_{4}}{c\omega_3} \leq \frac{a\omega_3+b(\sqrt{2}-1)\omega_3}{c\omega_3} =
				\end{cases} \\ 
				& = \begin{cases}
					\begin{cases}
						= \frac{a\sqrt{2}\omega_3+(b-a)\omega_{4}}{c\omega_3} \overbrace{\leq}^{\omega_{4} \leq \omega_3} \frac{a\sqrt{2}\omega_3+(b-a)\omega_3}{c\omega_3} = \frac{a(\sqrt{2}-1)+b}{\sqrt{a^2+b^2}} \\
						= \frac{b\sqrt{2}\omega_3+(a-b)\omega_{2}}{c\omega_3} \overbrace{\leq}^{\omega_{2} \leq \omega_3} \frac{b\sqrt{2}\omega_3+(a-b)\omega_3}{c\omega_3} = \frac{b(\sqrt{2}-1)+a}{\sqrt{a^2+b^2}}
					\end{cases} \\
					= \frac{a(\sqrt{2}-1)+b}{\sqrt{a^2+b^2}} \\
					= \frac{a+b(\sqrt{2}-1)}{\sqrt{a^2+b^2}}
				\end{cases} \leq \frac{2}{\sqrt{2+\sqrt{2}}},
    		\end{align*}
    		where the last inequality in the four ratios is obtained by maximization over the values of~$ a, b \in (0,1] $.
		\end{proof}
	
	Finally, we have all the pieces to prove our main result.

	\begin{theorem}
		\label{thm:4}
		In $ G_{8\text{corner}} $, $ \frac{\lVert \mathit{SGP_w}(s, t)\rVert}{\lVert \mathit{SP_w}(s,t) \rVert} \leq  \frac{2}{\sqrt{2+\sqrt{2}}} \approx 1.08 $.
	\end{theorem}

	\begin{proof}
		Let $ \mathit{SP_w}(s,t) $ be a weighted shortest path between two vertices $ s $ and $ t $ of a square mesh. Let~$ X(s,t) $ be the crossing path from~$ s $ to~$ t $ obtained from~$ \mathit{SP_w}(s,t) $. By Observation~\ref{thm:1}, $ \frac{\lVert X(s,t)\rVert}{\lVert \mathit{SP_w}(s,t)\rVert} \leq \frac{\lVert X(u_j,u_{j+1})\rVert}{\lVert \mathit{SP_w}(u_j,u_{j+1})\rVert} $, over all pairs~$ (u_j,u_{j+1}) $ of consecutive points where $ \mathit{SP_w}(s,t) $ and $ X(s,t) $ coincide.
		
		By Observations~\ref{obs:reduce0} and \ref{obs:reduce3}, and Lemma~\ref{lem:16}, the ratio $ \frac{\lVert X(u_j,u_{j+1})\rVert}{\lVert \mathit{SP_w}(u_j,u_{j+1})\rVert} $ for $ \mathit{SP_w}(u_j,u_{j+1}) $ and $ X(u_j,u_{j+1}) $ bounding polygons of type~$ P_0^1 $, $ P_1^\ell $, for $ \ell \geq 2 $, and~$ P_2^\ell $, for $ \ell \geq 1 $, does not generate the worst-case upper bound on the ratio $ \frac{\lVert X(s,t)\rVert}{\lVert \mathit{SP_w}(s,t)\rVert} $. Therefore, we only need to find an upper bound for polygons of type $ P_1^1$. According to Lemma~\ref{lem:equallength2square}, we know that if $ \mathit{SP_w}(s,t) $ and $ X(s,t) $ are the paths bounding polygons of type~$ P^1_1 $ that maximize the ratio $ \frac{\lVert X(u_j,u_{j+1})\rVert}{\lVert \mathit{SP_w}(u_j,u_{j+1})\rVert} $, then $ \lVert X(s,t)\rVert = \lVert \Pi_3^1(s,t)\rVert $. And, by Lemma~\ref{lem:19}, in this case we get that $ \frac{\lVert X(u_j,u_{j+1})\rVert}{\lVert \mathit{SP_w}(u_j,u_{j+1})\rVert} \leq \frac{2}{\sqrt{2+\sqrt{2}}} $.
		
		All this implies that $ \frac{\lVert X(s,t)\rVert}{\lVert \mathit{SP_w}(s,t)\rVert} $ is at most $ \frac{2}{\sqrt{2+\sqrt{2}}} $. Since $ \lVert \mathit{SGP_w}(s,t)\rVert \leq \lVert \mathit{X}(s,t)\rVert $, we have that  $ \frac{\lVert \mathit{SGP_w}(s,t)\rVert}{\lVert \mathit{SP_w}(s,t)\rVert} \leq \frac{2}{\sqrt{2+\sqrt{2}}} $.
	\end{proof}
	
	\subsection{Ratios $ \frac{\lVert \mathit{SGP_w}(s,t)\rVert}{\lVert \mathit{SVP_w}(s,t)\rVert} $ in $ G_{8\text{corner}} $ and $ \frac{\lVert \mathit{SVP_w}(s,t)\rVert}{\lVert \mathit{SP_w}(s,t)\rVert} $ for square cells}\label{sec:svpsquares}
	
	In this section we provide results for the ratios of the weighted shortest vertex path $ \mathit{SVP_w}(s,t) $, i.e., $ \frac{\lVert \mathit{SGP_w}(s,t)\rVert}{\lVert \mathit{SVP_w}(s,t)\rVert} $ and $ \frac{\lVert \mathit{SVP_w}(s,t)\rVert}{\lVert \mathit{SP_w}(s,t)\rVert} $. The length of a weighted shortest vertex path $ \mathit{SVP_w}(s,t) $ is an upper bound for the length of a weighted shortest path $ \mathit{SP_w}(s,t) $, so the upper bound on the ratio $ \frac{\lVert \mathit{SGP_w}(s, t)\rVert}{\lVert \mathit{SP_w}(s,t) \rVert} $ obtained in Theorem~\ref{thm:4} is an upper bound for $ \frac{\lVert \mathit{SGP_w}(s, t)\rVert}{\lVert \mathit{SVP_w}(s,t) \rVert} $. As a consequence, we obtain Corollary \ref{cor:3}.
	
	\begin{corollary}
		\label{cor:3}
		In $ G_{8\text{corner}}, \ \frac{\lVert \mathit{SGP_w}(s, t)\rVert}{\lVert \mathit{SVP_w}(s,t) \rVert} \leq \frac{2}{\sqrt{2+\sqrt{2}}} \approx 1.08 $.
	\end{corollary}
	
	When the weights of the cells are in the set $ \{1, \infty\} $, a worst-case lower bound on the ratio $ \frac{\lVert \mathit{SGP_w}(s, t)\rVert}{\lVert \mathit{SVP_w}(s,t) \rVert} $ was proved to be $ \frac{2}{\sqrt{2+\sqrt{2}}} $ by Nash~\cite{Nash}. Thus, for general (non-negative) weights this value is a worst-case lower bound on the ratio $ \frac{\lVert \mathit{SGP_w}(s, t)\rVert}{\lVert \mathit{SVP_w}(s,t) \rVert} $ for $ G_{8\text{corner}} $.
	
	As a corollary of Theorem~\ref{thm:4}, we also obtain Corollary~\ref{cor:5}. The result comes from the fact that $ \lVert \mathit{SVP_w}(s,t) \rVert $ is a lower bound for $ \lVert \mathit{SGP_w}(s,t) \rVert $ since every grid path is a vertex path. Recall that $ \mathit{SVP_w}(s, t) $ does not use $ G_{8\text{corner}} $, but $ G_{\text{corner}} $.
	
	\begin{corollary}
		\label{cor:5}
		In a square tessellation, $ \frac{\lVert \mathit{SVP_w}(s, t)\rVert}{\lVert \mathit{SP_w}(s,t) \rVert} \leq \frac{2}{\sqrt{2+\sqrt{2}}} $.
	\end{corollary}
	
	\begin{figure}[tb]
	    \centering
	    \includegraphics{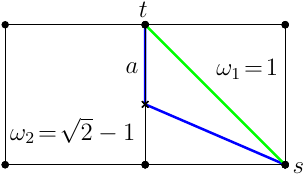}
	    \caption{$ \mathit{SP_w}(s,t) $ and $ \mathit{SVP_w}(s,t) $ are depicted in blue and green, respectively. The ratio~$ \frac{\lVert \mathit{SVP_w}(s,t)\rVert}{\lVert \mathit{SP_w}(s,t)\rVert} $ is $ \frac{\sqrt{2}\sqrt{\sqrt{2}-1}}{(\sqrt{2}-1)^{\frac{3}{2}}-\sqrt{2}+2} $ when $ a = \sqrt{\frac{3-2\sqrt{2}}{2\sqrt{2}-2}}\approx 0.46 $ and side length of mesh cell is~$ 1$.}
	    \label{fig:31}
    \end{figure}
	
	Finally, we provide a lower bound on the ratio $ \frac{\lVert \mathit{SVP_w}(s, t)\rVert}{\lVert \mathit{SP_w}(s,t) \rVert} $. The green path in Figure~\ref{fig:31} is a weighted shortest vertex path~$ \mathit{SVP_w}(s,t) $ between vertices $ s $ and $ t $, thus, we have the following result.
	
	\begin{observation}
	    \label{obs:5}
	    In a square tessellation, $ \frac{\lVert \mathit{SVP_w}(s, t)\rVert}{\lVert \mathit{SP_w}(s,t) \rVert} \geq  \frac{\sqrt{2}\sqrt{\sqrt{2}-1}}{(\sqrt{2}-1)^{\frac{3}{2}}-\sqrt{2}+2} \approx 1.07 $.
	\end{observation}
	
    \subsection{Ratios in $ G_{4\text{corner}} $}

        In this section we turn our attention to obtaining upper bounds in $ G_{4\text{corner}} $, i.e., a square mesh where each vertex is connected to its $ 4 $ adjacent vertices.
	
	The techniques we used in Sections~\ref{sec:8sq} and \ref{sec:svpsquares} for the $ G_{8\text{corner}} $ graph lead in a straightforward way to analogous results in this section for the ratios where the grid paths are paths in the $ G_{4\text{corner}} $ graph. Thus, for the sake of simplicity of exposition we just highlight the details that are different.
	
	The crossing path we defined in Definition~\ref{def:crossingsquare} is also a grid path in $ G_{4\text{corner}} $. Hence, almost all the tools in Section~\ref{sec:8sq} can still be used to upper-bound the ratios for paths intersecting $ G_{4\text{corner}} $. Since the crossing path in Definition~\ref{def:crossingsquare} is a valid path in $ G_{4\text{corner}} $, the polygons induced by the union of $ \mathit{SP_w}(s, t) $ and $X(s, t) $ are the same as those in Definitions~\ref{def:weakly1square} and \ref{def:weakly3square}. By Observations~\ref{obs:reduce0} and \ref{obs:reduce3}, and Lemma~\ref{lem:16}, an upper bound on the ratio $ \frac{\lVert X(u_j, u_{j+1})\rVert}{\lVert \mathit{SP_w}(u_j, u_{j+1}) \rVert} $ is attained by the paths bounding a polygon of type $ P^1_1 $.
    
    The first result that differs from the $ G_{8\text{corner}} $ case is an upper bound on polygons of type~$ P_1^1 $. The main difference is that now we cannot use shortcut paths~$ \Pi_i^1(s,t) $ from Definition~\ref{def:shortcut2square}. This is because the diagonal edge cutting across a cell (like the edge from $ v_1^3 $ to $ v_3^3 $ in Figure~\ref{fig:shortcutP2square}) is in $ G_{8\text{corner}} $ but {\em not} in $ G_{4\text{corner}} $. This increases the upper bound on the ratio. In fact, it also highlights the importance of the diagonal edges in $ G_{8\text{corner}} $.

		\begin{lemma}
			\label{lem:22}
			Let $ u_j, u_{j+1} \in S_i $ be two consecutive points where $ \mathit{SP_w}(s,t) $ and $ X(s,t) $ coincide in a square mesh $ \mathcal{S} $. If $ u_j $ and $ u_{j+1} $ induce a polygon of type~$ P^1_1 $, then $ \frac{\lVert X(u_j, u_{j+1})\rVert}{\lVert \mathit{SP_w}(u_j, u_{j+1}) \rVert} \leq \sqrt{2} $.
		\end{lemma}
		
		\begin{proof}
			Let $ (v^i_1, v^i_2, v^i_3, v^i_4) $ be the sequence of consecutive vertices on the boundary of $ S_i $ in clockwise order. Since $ u_j $ and $ u_{j+1} $ induce a polygon of type~$ P_1^1 $, $ \mathit{SP_w}(s,t) $ enters $ S_i $ from cell $ S_{i-1} $ through $ u_j $, and $ \mathit{SP_w}(s,t) $ leaves~$ S_i $ and enters cell $ S_{i+1} $ through~$ u_{j+1} $, see Figure~\ref{fig:54}.
			Suppose, without loss of generality, that $ u_j\in [v^i_1, v^i_2) $ and $u_{j+1} \in (v^i_2, v^i_3] $.

            \begin{figure}[htb]
			\centering
			\includegraphics{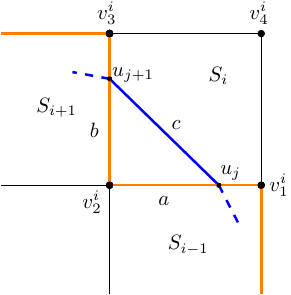}
    		\caption{Polygon of type~$ P_1^1 $ intersecting cell~$ S_i $, showing subpaths of the crossing path~$ X(s,t) $ (orange), and $ \mathit{SP_w}(s,t) $ (blue).}
    		\label{fig:54}
		\end{figure}
   
	    	Let $ a, b, c $ be the lengths $ |u_jv^i_2|, |v^i_2u_{j+1}| $, and $ |u_ju_{j+1}| $, respectively. According to Observation~\ref{obs:lengthsquares}, part~$2$, $c = \sqrt{a^2+b^2}$. Then, we want to upper-bound the ratio $ \frac{\lVert X(u_j, u_{j+1})\rVert}{\lVert \mathit{SP_w}(u_j, u_{j+1})\rVert} $ for all weight assignments $ \omega_{i-1}, \omega_i $ and $ \omega_{i+1} $.
			
			\begin{equation*}
				\frac{\lVert X(u_j, u_{j+1})\rVert}{\lVert \mathit{SP_w}(u_j, u_{j+1})\rVert} = \frac{a\min\{\omega_{i-1}, \omega_i\}+b\min\{\omega_i, \omega_{i+1}\}}{c\omega_i} \leq \frac{a\omega_i+b\omega_i}{c\omega_i} = \frac{a+b}{\sqrt{a^2+b^2}} \leq \sqrt{2},
    		\end{equation*}
    		where the last inequality is obtained by maximizing over the values of $ a, b \in (0,1] $. Also, in the last equality we assume that $ \omega_i \neq 0 $. However, if~$ \omega_i=0$, then the weighted length of $ X(u_j, u_{j+1}) $ and $ \mathit{SP_w}(u_j, u_{j+1}) $ is the same, and $ \frac{\lVert X(u_j, u_{j+1})\rVert}{\lVert \mathit{SP_w}(u_j, u_{j+1})\rVert} = 1$.
		\end{proof}

		Finally, since $ \lVert \mathit{SGP_w}(s, t) \rVert \leq \lVert X(s, t)\rVert $, by using the mediant inequality (see Observation~\ref{thm:1}), we obtain our main result.
		
		\begin{theorem}
		    \label{thm:4square}
	    	In $ G_{4\text{corner}} $, $ \frac{\lVert \mathit{SGP_w}(s, t)\rVert}{\lVert \mathit{SP_w}(s,t) \rVert} \leq  \sqrt{2} $.
	    \end{theorem}
	    
	    As a consequence of Theorem~\ref{thm:4square}, we obtain an upper bound on the ratio~$ \frac{\lVert \mathit{SGP_w}(s, t)\rVert}{\lVert \mathit{SVP_w}(s,t) \rVert} $. The result is obtained by taking into account that $ \lVert \mathit{SP_w}(s, t) \rVert \leq \lVert \mathit{SVP_w}(s, t)\rVert $.
	    
	    \begin{corollary}
		    \label{cor:sgpsvp4}
		    In $ G_{4\text{corner}}, \ \frac{\lVert \mathit{SGP_w}(s, t)\rVert}{\lVert \mathit{SVP_w}(s,t) \rVert} \leq \sqrt{2} $.
	    \end{corollary}

	\section{Hexagonal meshes}
	
	This section is devoted to upper-bounding the ratios $ \frac{\lVert \mathit{SGP_w}(s, t)\rVert}{\lVert \mathit{SP_w}(s,t) \rVert}, \ \frac{\lVert \mathit{SGP_w}(s, t)\rVert}{\lVert \mathit{SVP_w}(s,t) \rVert} $ and $ \frac{\lVert \mathit{SVP_w}(s, t)\rVert}{\lVert \mathit{SP_w}(s,t) \rVert} $ in a hexagonal mesh for any pair of points $ s, t $. Recall that given a hexagonal mesh, we define two types of $k$-corner grid graphs, depending on vertex connectivity, namely the $ 3 $- and $ 12$-corner grid graphs, i.e., $ G_{3\text{corner}} $ and $ G_{12\text{corner}} $, see Section~\ref{cap.introduction}. Thus, we split the proof into two cases depending on the number of neighbors of each vertex. The techniques that we apply in this section are similar to the ones used in Section~\ref{cap.sq}. Hence, we just explain in detail the results that differ. In particular, omitted proofs are analogous to the corresponding ones in Section~\ref{sec:8sq}
	
	\subsection{$ \frac{\lVert \mathit{SGP_w}(s, t)\rVert}{\lVert \mathit{SP_w}(s,t) \rVert} $ ratio in a hexagonal mesh}
	
	Consider a hexagonal mesh $ \mathcal{H} $ where each cell is a closed hexagon of side length 1. Let $ (H_1,\ldots, H_n) $ be the ordered sequence of consecutive hexagonal cells intersected by a shortest path $ \mathit{SP_w}(s,t) $. Let~$ v^i_1, v^i_2, v^i_3, v^i_4, v^i_5, v^i_6 $ be the six corners on the boundary of the cell~$ H_i, \ 1 \leq i \leq n $, in clockwise order. Let $ (s = a_1, a_2, \ldots, a_{n+1} = t) $ be the sequence of consecutive points where~$ \mathit{SP_w}(s, t) $ changes the cell(s) it belongs to in $ \mathcal{H} $. In particular, let $ a_i $ and~$ a_{i+1} $ be, respectively, the points where $ \mathit{SP_w}(s,t) $ enters and leaves $ H_i $. Since the graphs~$ G_{3\text{corner}} $ and $ G_{12\text{corner}} $ defined on a hexagonal mesh are different to the graphs defined in a square mesh, we first define a crossing path $ X(s,t) $ in both graphs. If $ s, t \in G_{3\text{corner}} $, the crossing path $ X(s,t) $ is defined in the following way.

	
	\begin{definition}
		\label{def:crossinghex3}
		The crossing path $ X(s,t) $ between two vertices $ s $ and $ t $ in $ G_{3\text{corner}} $ is the path in $ \mathcal{H} $ with vertex sequence $ (X_1, \ldots, X_n) $, where $ X_i $ is a sequence of at most 4 vertices determined by the pair~$ (a_i, a_{i+1}), \ 1 \leq i \leq n $.
		\begin{enumerate}
		\setlength\itemsep{0em}
		
		    \item If $ a_i $ and $ a_{i+1} $ are on the same edge $ e_1^i \in H_i $, let $ u $ and $ v $ be the endpoints of $ e^i_1 $, where $ a_i $ is encountered before $ a_{i+1} $ when traversing $ e^i_1 $ from $ u $ to $ v $. If $ i=1$ then $ X_i = (a_1, v) $, otherwise $ X_i=(v) $, see Figure~\ref{fig:caseacrossingshex3}.
		    
		    \item If $ a_i $ and $ a_{i+1} $ belong to two adjacent edges $ e_1^i, e_2^i \in H_i $, let $ v $ be the corner of $ H_i $ shared by $ e_1^i $ and $ e_2^i $, and let $ u, u' $ be the other endpoints of $ e^i_1 $ and $ e^i_2 $, respectively. If $ a_i=u$ then $ X_i = (u, v, u') $, otherwise $ X_i = (v,u') $, see Figure~\ref{fig:casebcrossingshex3}.

                \item Let $ e_1^i = [u,v], e_2^i=[v,u'], e_3^i=[u',v'] $ be three consecutive edges of $ H_i$, where $ u $ is encountered before $ v$, and $ u' $ before $ v'$ when traversing $ e_1^i, e_2^i, e_3^i $ in that order. Let $ a_{i+1} \in e_3^i $. If $ a_i = u $ then $ X_i = (u,v,u',v')$, otherwise, if $ a_i $ belong to the interior of $ e_1^i $, $ X_i = (v,u',v') $, see Figure~\ref{fig:caseccrossingshex3}.
		    
		    \item If $ a_i $ and $ a_{i+1} $ belong to the interior of two parallel edges $ e_1^i $ and $ e_2^i $, and the last point in $ X_{i-1} $ was the endpoint of $ e^i_1 $ to the left (resp., right) of $ \overrightarrow{a_ia_{i+1}} $, $ X_i $ is the sequence of corners of $ H_i $, in clockwise (resp., counterclockwise) order to the left (resp., right) of $ \overrightarrow{a_ia_{i+1}} $, see Figure~\ref{fig:casedcrossingshex3}.
		\end{enumerate}
	\end{definition}
	
	\begin{figure}[tb]
		\captionsetup[sub]{justification=centering}
		\centering
		\begin{subfigure}[b]{0.3\textwidth}
	        \includegraphics{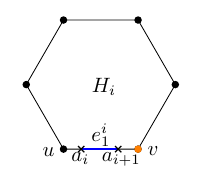}
	        \caption{}
	        \label{fig:caseacrossingshex3}
	    \end{subfigure}
    	\qquad\qquad
		\begin{subfigure}[b]{0.3\textwidth}
	        \includegraphics{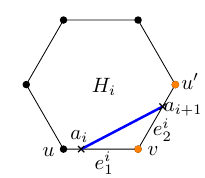}
	        \caption{}
	        \label{fig:casebcrossingshex3}
    	\end{subfigure}
            \\
            
		\begin{subfigure}[b]{0.3\textwidth}
	        \includegraphics{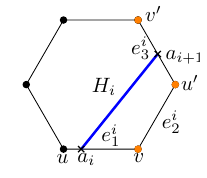}
	        \caption{}
	        \label{fig:caseccrossingshex3}
    	\end{subfigure}
    	\qquad\qquad
		\begin{subfigure}[b]{0.3\textwidth}
	        \includegraphics{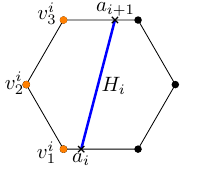}
	        \caption{}
	        \label{fig:casedcrossingshex3}
    	\end{subfigure}
    	\caption{The points $ a_i $ and $ a_{i+1} $ represent the intersection of the shortest path $ \mathit{SP_w}(s,t) $ (blue) with the hexagonal cell $ H_i $. The orange vertices represent the vertices of the crossing path $ X(s,t) \in G_{3\text{corner}} $.}
    \end{figure}
	
	Now, when $ s, t \in G_{12\text{corner}}$, some other edges are allowed, and the corresponding path $ X(s,t) $ is different.
	
	\begin{definition}
		\label{def:crossinghex12}
		The crossing path $ X(s,t) $ between two vertices $ s $ and $ t $ in $ G_{12\text{corner}} $ is the path in $ \mathcal{H} $ with vertex sequence $ (X_1, \ldots, X_n) $, where $ X_i $ is a sequence of at most 3 vertices determined by the pair~$ (a_i, a_{i+1}), \ 1 \leq i \leq n $.
		\begin{enumerate}
		\setlength\itemsep{0em}

                \item If $ a_i$ and $ a_{i+1} $ are two corners of $ H_i$, $ X_i=(a_i,a_{i+1})$, see Figure~\ref{fig:caseacrossingshex12}.
		
		    \item If $ a_i $ and $ a_{i+1} $ are on the same edge $ e_1^i \in H_i $, let $ u $ and $ v $ be the endpoints of $ e^i_1 $, where $ a_i $ is encountered before $ a_{i+1} $ when traversing $ e^i_1 $ from $ u $ to $ v $. If $ i=1$ then $ X_i = (a_1, v) $, otherwise $ X_i=(v) $.

                \item If $ a_i $ and $ a_{i+1} $ belong to two adjacent edges $ e_1^i, e_2^i \in H_i $, let $ v $ be the corner of $ H_i $ shared by $ e_1^i $ and $ e_2^i $, and let $ u, u' $ be the other endpoints of $ e^i_1 $ and $ e^i_2 $, respectively. If $ a_i=u$ then $ X_i = (u, v, u') $, otherwise $ X_i = (v,u') $, see Figure~\ref{fig:casebcrossingshex12}.

                \item Let $ e_1^i = [u,v], e_2^i=[v,u'], e_3^i=[u',v'] $ be three consecutive edges of $ H_i$, where $ u $ is encountered before $ v$, and $ u' $ before $ v'$ when traversing $ e_1^i, e_2^i, e_3^i $ in that order. Let $ a_{i+1} \in e_3^i $. If $ a_i = u $ then $ X_i = (u,u',v')$, see Figure~\ref{fig:caseccrossingshex12}. otherwise, if $ a_i $ belong to the interior of $ e_1^i $, $ X_i = (u,v,u',v') $.

                \item If $ a_i $ and $ a_{i+1} $ belong to the interior of two parallel edges $ e_1^i $ and $ e_2^i $, and the last point in $ X_{i-1} $ was the endpoint of $ e^i_1 $ to the left (resp., right) of $ \overrightarrow{a_ia_{i+1}} $, $ X_i $ is the endpoint of $ e^i_1 $ to the left (resp., right), see Figure~\ref{fig:casedcrossingshex12}.
		\end{enumerate}
	\end{definition}
	
	\begin{figure}[tb]
		\captionsetup[sub]{justification=centering}
		\centering
		\begin{subfigure}[b]{0.3\textwidth}
			\centering
	        \includegraphics{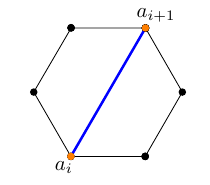}
	        \caption{}
	        \label{fig:caseacrossingshex12}
	    \end{subfigure}
    	\qquad
		\begin{subfigure}[b]{0.3\textwidth}
   		    \centering
	        \includegraphics{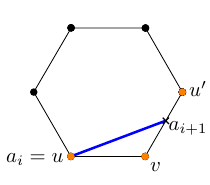}
	        \caption{}
	        \label{fig:casebcrossingshex12}
    	\end{subfigure}
    	\qquad
		\begin{subfigure}[b]{0.3\textwidth}
   		    \centering
	        \includegraphics{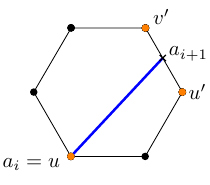}
	        \caption{}
	        \label{fig:caseccrossingshex12}
    	\end{subfigure}
    	\qquad
		\begin{subfigure}[b]{0.3\textwidth}
   		    \centering
	        \includegraphics{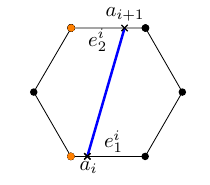}
	        \caption{}
	        \label{fig:casedcrossingshex12}
    	\end{subfigure}
    	\caption{The points $ a_i $ and $ a_{i+1} $ represent the intersection of the shortest path $ \mathit{SP_w}(s,t) $ (blue) with an hexagonal cell. The orange vertices represent the vertices of the crossing path $ X(s,t) \in G_{12\text{corner}} $.}
    \end{figure}
	
	Since we defined the path $ X(s,t) $ as a grid path, its weighted length is as least as large as the weighted length of the shortest grid path $ \mathit{SGP_w}(s, t) $, and we can upper-bound the ratio $ \frac{\lVert \mathit{SGP_w}(s, t)\rVert}{\lVert \mathit{SP_w}(s,t) \rVert} $ by the ratio $ \frac{\lVert X(s, t)\rVert}{\lVert \mathit{SP_w}(s,t) \rVert} $. Let $ (s=u_1, u_2, \ldots, u_\ell=t) $ be the ordered sequence of consecutive points where $ X(s,t)$ and $ \mathit{SP_w}(s,t) $ coincide. Once we have defined the crossing paths in each of the graphs, and using Observation~\ref{thm:1}, we want to upper-bound the ratio $ \frac{\lVert X(u_j, u_{j+1})\rVert}{\lVert \mathit{SP_w}(u_j,u_{j+1}) \rVert} $, for any pair $ (u_j, u_{j+1}), \ j \in \{1, \ldots, \ell-1\} $. So, having a look at the union of $ X(s, t) $ and $ \mathit{SP_w}(s,t) $ between consecutive intersection points $ u_j$ and $ u_{j+1} $, we notice that certain polygons arise. The type of polygon depends on the number of different cells that are crossed by the subpaths $ \mathit{SP_w}(u_j,u_{j+1}) $ and $X(u_j, u_{j+1})$. If these two subpaths intersect the interior of at most one cell, we have the types of polygons in Definition~\ref{def:weakly1hex3}.
	
    \begin{definition}
		\label{def:weakly1hex3}
		Let $ u_j, u_{j+1} \in H_i $ be two consecutive points where $ \mathit{SP_w}(s, t) $ and~$ X(s, t) \in G_{3\text{corner}} $ coincide. A polygon induced by $ u_j $ and $ u_{j+1} $ is of \emph{type}~$ P^1_k, \ 0 \leq k \leq 3 $ if $ \mathit{SP_w}(u_j, u_{j+1}) \cup X(u_j, u_{j+1}) $ intersects the interior of~$ k+1 $ different edges of $ H_i $.
	\end{definition}
	
	\begin{figure}[htb]
		\centering
		\includegraphics[scale=0.9]{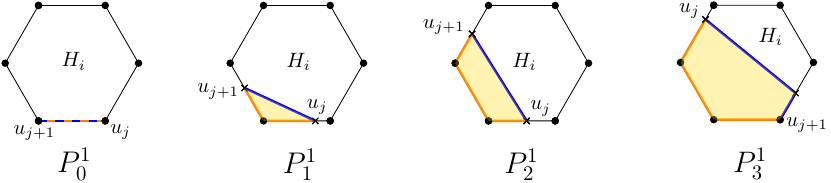}
		\captionof{figure}{Some polygons of type~$ P^1_k, k\in\{0,\ldots,3\} $, bounded by the subpaths $ X(u_j,u_{j+1}) $ (orange) and $ \mathit{SP_w}(u_j,u_{j+1}) $ (blue) in $ G_{3\text{corner}} $.}
		\label{fig:onecellhex3}
	\end{figure}

        Otherwise, if the subpaths $ \mathit{SP_w}(u_j,u_{j+1}) $ and $X(u_j, u_{j+1})$ intersect the interior of more than one cell, we obtain the polygons in Definition~\ref{def:weakly3hex3}.
	
	\begin{definition}
		\label{def:weakly3hex3}
		Let $ u_j\in e_1^i $ and $ u_{j+1} $ be two consecutive points where $ \mathit{SP_w}(s, t) $ and~$ X(s, t) \in G_{3\text{corner}} $ coincide. Let $ u_{j+1} $ be to the right (resp., left) of the line through~$ u_j $ perpendicular to $e_1^i$, with respect to the $ SP_w(s,t) $, when considering $ SP_w(s,t) $ oriented from $ s $ to $ t $, see Figure~\ref{fig:morecellhex3}. A polygon induced by $ u_j $ and $ u_{j+1} $ is of \emph{type}~$ P^{\ell}_k, \ \ell \geq 2, \ k \in \{1, 2, 3\} $, if:
		\begin{itemize}
		\setlength\itemsep{0em}
		    \item $ \mathit{SP_w}(u_j, u_{j+1}) $ intersects the interior of $ \ell $ consecutive cells $ H_i, \ldots, H_m$, with $ \ell=m-i+1$.
			\item $ X(u_j, u_{j+1}) $ contains all the vertices of the $ \ell $ cells intersected by $ \mathit{SP_w}(u_j, u_{j+1}) $ that are to the right (resp., left) of $ \mathit{SP_w}(u_j, u_{j+1}) $.
			\item $ \mathit{SP_w}(u_j, u_{j+1}) \cup X(u_j, u_{j+1}) $ intersects the interior of~$k+1 $ different edges of~$ S_m$.
		\end{itemize}
	\end{definition}
	
	\begin{figure}[htb]
		\centering
		\includegraphics[scale=0.9]{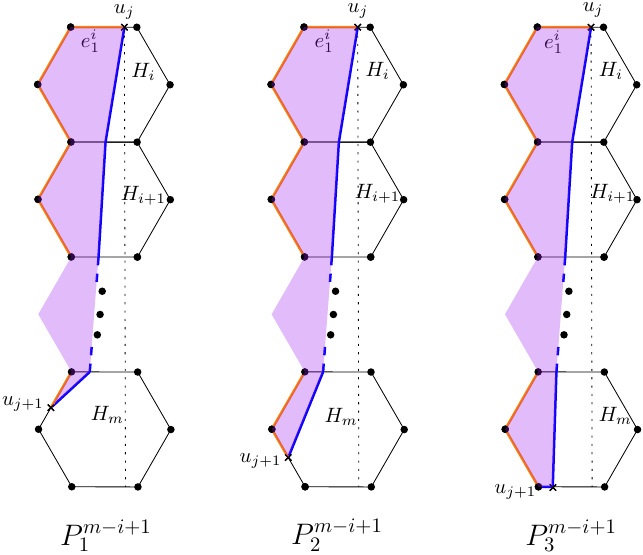}
		\captionof{figure}{Three polygons of type~$ P^{m-i+1}_1 $, $ P_2^{m-i+1} $ and~$ P^{m-i+1}_3 $, bounded by the subpaths~$ \mathit{SP_w}(u_j,u_{j+1}) $ (blue) and $ X(u_j, u_{j+1}) $ (orange) in $ G_{3\text{corner}} $.}
		\label{fig:morecellhex3}
	\end{figure}
	
	Note that in $ G_{12\text{corner}} $ it is possible for the crossing path $ X(s, t) $ to travel through some chords of the hexagonal cells. Thus, we need to define other types of polygons in this graph. Definition~\ref{def:weakly1hex12} gives us the types of polygons when $ u_j $ and $ u_{j+1} $ belong to the same cell.
	
	\begin{definition}
		\label{def:weakly1hex12}
		Let $ u_j, u_{j+1} \in H_i $ be two consecutive points where $ \mathit{SP_w}(s, t) $ and~$ X(s, t) \in G_{12\text{corner}} $ coincide. A polygon induced by $ u_j $ and $ u_{j+1} $ is of \emph{type}~$ P^1_0 $ if $ \mathit{SP_w}(u_j, u_{j+1}) $ and $ X(u_j, u_{j+1}) $ coincide.
  
        Let $ e^i_1 $ be the edge containing $ u_j $. A polygon induced by $ u_j $ and $ u_{j+1} $ is of \emph{type}~$ P^1_3 $ if the only edges of $ H_i $ that $ \mathit{SP_w}(u_j, u_{j+1}) $ and $ X(u_j, u_{j+1}) $ intersect are $ e^i_1 $ and its parallel edge in~$ H_i $.
			
	\end{definition}

        Note that in the case where $ u_j $ and $ u_{j+1}$ belong to the same cell, the type of polygon $ P_3^1$ differs from the one in Definition~\ref{def:weakly1hex3}, see Figure~\ref{fig:onecellhex12}. In addition, $ P_0^1 $ is also different since now we do not need $ u_j $ and $ u_{j+1}$ to be two adjacent vertices. We now define the types of polygons obtained when $ \mathit{SP_w}(u_j, u_{j+1}) $ and $ X(u_j, u_{j+1}) $ intersect the interior of more than one cell.
	
	\begin{figure}[htb]
		\centering
		\includegraphics[scale=0.9]{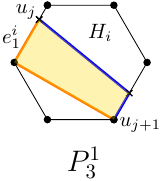}
		\captionof{figure}{Polygon of type~$ P^1_3 $, bounded by the subpaths $ X(u_j,u_{j+1}) $ (orange) and $ \mathit{SP_w}(u_j,u_{j+1}) $ (blue) in $ G_{12\text{corner}} $.}
		\label{fig:onecellhex12}
	\end{figure}
	
	\begin{definition}
		\label{def:weakly3hex12}
		Let $ u_j \in e_1^i$ and $ u_{j+1} $ be two consecutive points where~$ \mathit{SP_w}(s, t) $ and~$ X(s, t) \in G_{12\text{corner}} $ coincide. Let $ u_{j+1} $ be to the right (resp., left) of the line through~$ u_j $ perpendicular to $e_1^i$, with respect to the $ SP_w(s,t) $, when considering $ SP_w(s,t) $ oriented from $ s $ to $ t $, see Figure~\ref{fig:morecellhex12}. A polygon induced by $ u_j $ and $ u_{j+1} $ is of \emph{type}~$ P^{\ell}_k, \ \ell \geq 2, \ k \in \{1, 2, 3\} $, if:
		\begin{itemize}
		\setlength\itemsep{0em}
		    \item $ \mathit{SP_w}(u_j, u_{j+1}) $ intersects the interior of $ \ell $ consecutive cells $ H_i, \ldots, H_m$, with $ \ell=m-i+1$.
			\item $ X(u_j, u_{j+1}) $ contains the right (resp., left) endpoint of all the edges intersected by $ \mathit{SP_w}(u_j, u_{j+1}) $, with respect to $ \mathit{SP_w}(u_j, u_{j+1}) $, when considering $ \mathit{SP_w}(u_j, u_{j+1}) $ oriented from $ u_j $ to $ u_{j+1} $.
			\item If $ k=1,2$, $ \mathit{SP_w}(u_j, u_{j+1}) \cup X(u_j, u_{j+1}) $ intersects the interior of~$k+1 $ different edges of $ H_m $. If $ k=3$, $ \mathit{SP_w}(u_j, u_{j+1}) $ intersects the interior of the two edges in $ H_m $ parallel to $ e_1^i $.
		\end{itemize}
	\end{definition}
	
	\begin{figure}[htb]
		\centering
		\includegraphics[scale=0.9]{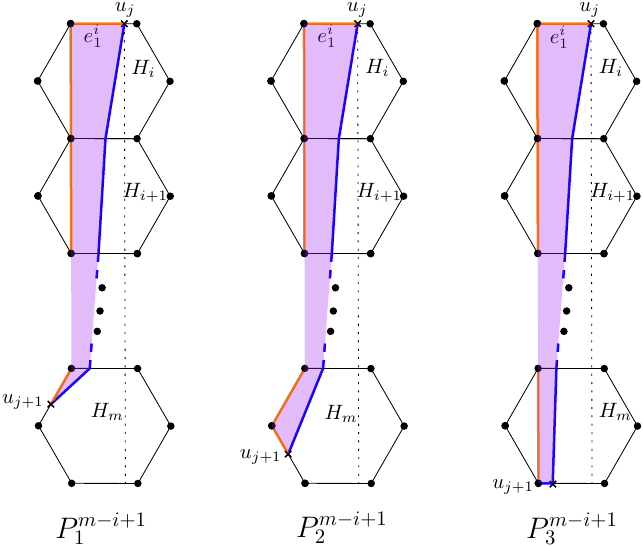}
		\captionof{figure}{Three polygons of type~$ P^{m-i+1}_1 $, $ P_2^{m-i+1} $ and~$ P^{m-i+1}_3 $, bounded by the subpaths~$ \mathit{SP_w}(u_j, u_{j+1}) $ (blue) and $ X(u_j, u_{j+1}) $ (orange) in $ G_{12\text{corner}} $.}
		\label{fig:morecellhex12}
	\end{figure}
	
	Using arguments analogous to those in Propositions~\ref{prop:onlyvertex} and \ref{prop:unique}, one can prove that if the crossing points $u_j,u_{j+1}$ belong to the same cell, we obtain the polygons defined in Definitions~\ref{def:weakly1hex3} and \ref{def:weakly1hex12} (see Figures~\ref{fig:onecellhex3} and \ref{fig:onecellhex12}, for $ X(s,t) $ contained in $ G_{3\text{corner}} $ and $ G_{12\text{corner}} $). Moreover, if the crossing points belong to different cells, the only polygons that arise are those in Definitions~\ref{def:weakly3hex3} and \ref{def:weakly3hex12} (see Figures~\ref{fig:morecellhex3} and \ref{fig:morecellhex12}, respectively).

    Observations~\ref{obs:reduce0} and \ref{obs:reduce3} can also be generalized when the cells of the tessellation are hexagons. Hence, the only relevant ratios are those where the shortest path and the crossing path bound polygons of type~$ P_1^1 $, $ P_2^1$, $ P_1^\ell$ and~$ P_2^\ell$ in both graphs~$ G_{3\text{corner}} $ and~$ G_{12\text{corner}}$. In the remainder of this section, we prove an upper bound for the weighted length of the crossing path with respect to the weighted length of a shortest path when bounding these four types of polygons in each of the graphs.
	
	The following two lemmas provide upper bounds on the ratio $ \frac{\lVert X(u_j, u_{j+1})\rVert}{\lVert \mathit{SP_w}(u_j, u_{j+1}) \rVert} $ when $ \mathit{SP_w}(u_j, u_{j+1}) $ bounds polygons of type~$ P_1^1 $, $ P_2^1$, and~$ X(s,t) \in G_{3\text{corner}} $. The proofs are similar to the proof of Lemma~\ref{lem:22} for squares.
	
	\begin{lemma}
		\label{lem:9}
		Let $ u_j, \ u_{j+1} \in \mathcal{H} $ be two consecutive points where $ \mathit{SP_w}(s,t) $ and $ X(s,t) \in G_{3\text{corner}} $ coincide. If $ u_j $ and $ u_{j+1} $ induce a polygon of type~$ P_1^1$, then $ \frac{\lVert X(u_j, u_{j+1})\rVert}{\lVert \mathit{SP_w}(u_j, u_{j+1}) \rVert} \leq \frac{2}{\sqrt{3}} \approx 1.15 $.
	\end{lemma}
	
	\begin{lemma}
		\label{lem:10}
		Let $ u_j, \ u_{j+1} \in \mathcal{H} $ be two consecutive points where $ \mathit{SP_w}(s,t) $ and $ X(s,t) \in G_{3\text{corner}} $ coincide. If $ u_j $ and $ u_{j+1} $ induce a polygon of type~$ P_2^1$, then $ \frac{\lVert X(u_j, u_{j+1})\rVert}{\lVert \mathit{SP_w}(u_j, u_{j+1}) \rVert} \leq \frac{3}{2} $.
	\end{lemma}
	
	Before obtaining the upper bounds when the shortest path and the crossing path bound the remaining types of polygons, we need to adapt Definition~\ref{def:Ptriples} of triples in a square mesh to obtain four new types of triples on a hexagonal mesh. These triples are valid for both graphs~$ G_{3\text{corner}} $ and~$ G_{12\text{corner}} $. However, for $ G_{3\text{corner}} $ we just need the triples containing polygons of type~$ P_1^\ell $ and~$ P_2^\ell$, while for $ X(s,t) \in G_{12\text{corner}} $ we need a triple for each type of polygon.
		
        \begin{definition}
	        \label{def:Ptriples12}
	        A $ P_k^\ell$-triple, for $ k\in\{1,2\}, \ell \geq 1$, from a vertex $ s $ to a vertex $ t $ is defined as a set of $ \ell+4$ cells $ H_1, \ldots, H_{\ell+4}$ with the following properties:
	        \begin{itemize}
                    \setlength\itemsep{0em}
	            \item $ s$ is the vertex common to $ H_1$ and $ H_2$, and not adjacent to $ H_3$.
	            \item $ t$ is the vertex common to $ H_{\ell+3}$ and $ H_{\ell+4}$, and not adjacent to~$ H_{\ell+2}$.
	            \item The union of $ \mathit{SP_w}(s, t) $ and $ X(s,t) $ determines two polygons of type~$ P_1^1$, and one polygon of type~$ P_k^\ell$ in between. 
	        \end{itemize}
	    \end{definition}
	    
	    Figures~\ref{fig:hexagonaltriples} and \ref{fig:hexagonaltriples2} depict four of the six possible types of triples in Definition~\ref{def:Ptriples12}. Figures \ref{fig:hextriple3corner2} and \ref{fig:hextriple3corner1} represent the triples in $ G_{3\text{corner}} $ (analogous triples can be defined for $ X(s,t) \in G_{12\text{corner}} $), and Figures \ref{fig:hextriple12corner1} and \ref{fig:hextriple12corner2}, two triples in~$ G_{12\text{corner}} $.
	    
	    \begin{figure}[tb]
		\centering
		\begin{subfigure}[b]{0.4\textwidth}
   			\centering
	    	\includegraphics[scale=0.7]{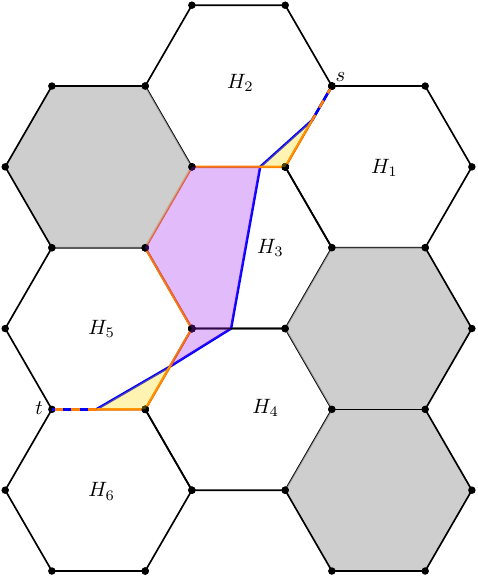}
			\caption{$P_1^2$-triple from $ s $ to $ t $ is represented in white.}
			\label{fig:hextriple3corner2}
    	\end{subfigure}
    	\qquad\quad
		\begin{subfigure}[b]{0.4\textwidth}
			\centering
	    	\includegraphics[scale=0.7]{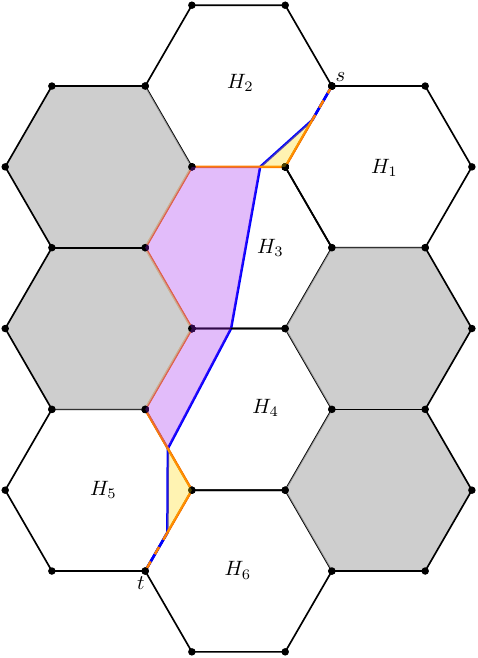}
	    	\caption{$P_2^2$-triple from $ s $ to $ t $ is represented in white.}
	    	\label{fig:hextriple3corner1}
	    \end{subfigure}
    	\caption{The crossing path $ X(s,t) $ (orange) and $ \mathit{SP_w}(s,t) $ (blue) traversing some triples in~$ G_{3\text{corner}} $.}
    	\label{fig:hexagonaltriples}
	\end{figure}

        \begin{figure}[tb]
		\centering
		\begin{subfigure}[t]{0.4\textwidth}
			\centering
	    	\includegraphics[scale=0.7]{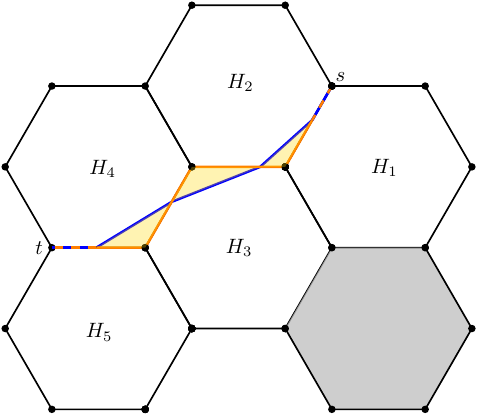}
	    	\caption{$P_1^1$-triple between $ s $ and $t$.}
	    	\label{fig:hextriple12corner1}
	    \end{subfigure}
	    \qquad\quad
		\begin{subfigure}[t]{0.4\textwidth}
			\centering
	    	\includegraphics[scale=0.7]{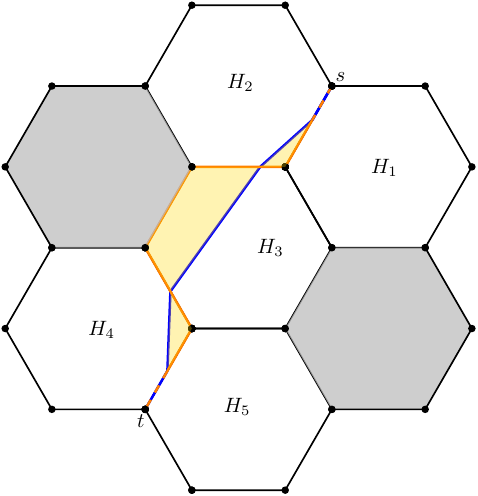}
	    	\caption{$P_2^1$-triple between $ s $ and $t$.}
	    	\label{fig:hextriple12corner2}
	    \end{subfigure}
    	\caption{The crossing path $ X(s,t) $ (orange) and $ \mathit{SP_w}(s,t) $ (blue) traversing some triples in $ G_{12\text{corner}} $.}
    	\label{fig:hexagonaltriples2}
	\end{figure}

     Now, we obtain a result analogous to Lemma~\ref{lem:16}, that is valid when $ X(s,t) \in G_{3\text{corner}} $ and $ X(s,t) \in G_{12\text{corner}}$. Note that this result is not possible when the cells are equilateral triangles because of the position of the cells and the definition of the crossing path, see \cite{bose2023approximating}.
	
	\begin{lemma}
		\label{lem:largecasehex}
    		The ratio $ \frac{\lVert \mathit{SGP_{w}}(s,t)\rVert}{\lVert \mathit{SP_{w}}(s,t)\rVert} $ in a $ P^\ell_k $-triple, $ k \in \{1,2\} $, is upper-bounded by the ratio $\frac{\lVert X(u_j, u_{j+1})\rVert}{\lVert \mathit{SP_w}(u_j, u_{j+1}) \rVert}$ when $ \mathit{SP_{w}}(u_j,u_{j+1}) $ and $ X(u_j,u_{j+1}) $ bound a polygon of type~$ P^1_k $, for any pair $ (u_j, u_{j+1}) $ of consecutive points where $ \mathit{SP_{w}}(s,t) $ and $ X(s,t) $ coincide.
	\end{lemma}
	
	Lemma~\ref{lem:largecasehex} implies that we just need to calculate an upper bound on the ratio $ \frac{\lVert X(u_j, u_{j+1})\rVert}{\lVert \mathit{SP_w}(u_j, u_{j+1}) \rVert} $ for $X(u_j, u_{j+1})$ and $ \mathit{SP_w}(u_j, u_{j+1})$ bounding polygons of type $ P^1_1 $ and $ P^1_2 $. According to Lemmas~\ref{lem:9} and  \ref{lem:10}, an upper bound on the ratio~$ \frac{\lVert X(u_j, u_{j+1})\rVert}{\lVert \mathit{SP_w}(u_j,u_{j+1}) \rVert} $ for a shortest path and a crossing path bounding these two types of polygons in $ G_{3\text{corner}} $ is~$ \frac{3}{2} $. Hence, using Observation~\ref{thm:1}, and the fact that~$ \lVert \mathit{SGP_w}(s, t)\rVert \leq \lVert X(s, t)\rVert$, we obtain the desired result.

	\begin{theorem}
		\label{thm:ratiohex3corner}
		In $ G_{3\text{corner}} $, $ \frac{\lVert \mathit{SGP_w}(s, t)\rVert}{\lVert \mathit{SP_w}(s,t) \rVert} \leq \frac{3}{2} $.
	\end{theorem} 
 
        We now turn our attention to $G_{12\text{corner}} $. In this case, as in $G_{8\text{corner}} $, the crossing path $ X(s,t) $ can intersect the chords of the cells. Thus, the sequence of vertices intersected by $ X(s,t) $ is not the same as in $G_{3\text{corner}} $. Thus, we need to define two classes of shortcut paths that intersect the edges of certain cells that~$ X(s,t)\in G_{12\text{corner}} $ does not intersect. Definition~\ref{def:8} determines these grid paths.

	\begin{definition}
		\label{def:8}
		Let $ (v^i_1, v^i_2, v^i_3, v^i_4, v^i_5, v^i_6) $ be the sequence of vertices of cell $ H_i $, in clockwise or counterclockwise order. Let $ \mathit{SP_w}(s,t) $ enter cell $ H_i $ through the edge $ [v_j^i, v_{\text{mod(j, 6)}+1}^i] $.
		\begin{itemize}
  \setlength\itemsep{0em}
		    \item If $ X(s,t) $ contains the subpath $ (v^i_{\text{mod(j, 6)}+1}, v^i_j, v^i_{\text{mod(j+4, 6)}+1}) $, the \emph{shortcut path $ \Pi_{i}^2(s,t) $} is defined as the grid path $ X(s,v^i_{\text{mod(j, 6)}+1}) \cup (v^i_{\text{mod(j, 6)}+1}, v^i_{\text{mod(j+4, 6)}+1}) \cup X(v^i_{\text{mod(j+4, 6)}+1},t) $, see yellow path in Figure~\ref{fig:27}. Symmetrically, if $ X(s,t) $ contains the subpath $ (v^i_j, v^i_{\text{mod(j, 6)}+1}, v^i_{\text{mod(j+1, 6)}+1}) $, $ \Pi_{i}^2(s,t) $ is defined as the grid path $ X(s,v^i_j) \cup (v^i_j, v^i_{\text{mod(j+1, 6)}+1}) \cup X(v^i_{\text{mod(j+1, 6)}+1},t) $.
                \item If $ X(s,t) $ contains the subpath $ (v^i_{\text{mod(j, 6)}+1}, v^i_j, v^i_{\text{mod(j+4, 6)}+1}, v_{\text{mod(j+3, 6)}+1}^i) $, the \emph{shortcut path $ \Pi_{i}^3(s,t) $} is defined as the grid path $ X(s,v^i_{\text{mod(j, 6)}+1}) \cup (v^i_{\text{mod(j, 6)}+1}, v^i_{\text{mod(j+3, 6)}+1}) \cup X(v^i_{\text{mod(j+3, 6)}+1},t) $, see green path in Figure~\ref{fig:27}. Symetrically, if $ X(s,t) $ contains the subpath $ (v^i_j, v^i_{\text{mod(j, 6)}+1}, v^i_{\text{mod(j+1, 6)}+1}, v_{\text{mod(j+2, 6)}+1}^i) $, $ \Pi_{i}^3(s,t) $ is defined as the grid path $ X(s,v^i_j) \cup (v^i_j, v^i_{\text{mod(j+2, 6)}+1}) \cup X(v^i_{\text{mod(j+2, 6)}+1},t) $.
		\end{itemize}
	\end{definition}
	
	\begin{figure}[htb]
		\centering
		\includegraphics[scale=0.7]{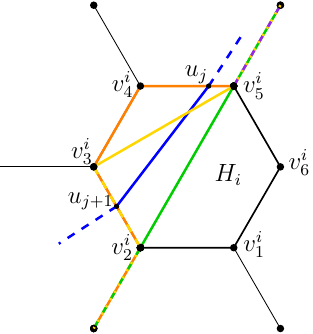}
	    	\caption{Subpaths of the grid paths $ \Pi_i^2(s,t) $ (yellow), $ \Pi_i^3(s,t) $ (green). The crossing path~$ X(s,t) $ (orange) and $ \mathit{SP_w}(s,t) $ (blue) intersect cell $ H_i $.}
    	\label{fig:27}
	\end{figure}

        In shortcut paths $ \Pi_i^2(s,t) $ and $ \Pi_i^3(s,t) $ there is an edge between two non-adjacent vertices of cell $ H_i $. Hence, they allow us to improve the worst-case upper bounds on the ratio $ \frac{\lVert X(s, t)\rVert}{\lVert \mathit{SP_w}(s, t) \rVert} $ by using a similar result to Lemma~\ref{lem:Ptriple}.
 
		Lemma~\ref{lem:Ptriple} states that a triple can be defined from a given polygon of certain type in a square mesh. This result can also be generalized to hexagonal meshes. Hence, we can assume that an upper bound on the ratio $ \frac{\lVert X(s, t)\rVert}{\lVert \mathit{SP_w}(s, t) \rVert} $ is given when the shortest path intersects either a $P_1^1$-triple or a $P_2^1$-triple. 
		Thus, using the grid paths specified in Definition~\ref{def:8}, and the following lemmas, we obtain a relation between the weights $ \omega_\ell $ of each cell in the triples where the crossing path $ X(s,t) $ and each of the grid paths from Definition~\ref{def:8} do not coincide.
		
		\begin{lemma}
			\label{lem:equallengthhex121}
			Let $ X(s,t) \in G_{12\text{corner}} $. Let $ \frac{\lVert \mathit{SGP_w}(s',t') \rVert}{\lVert \mathit{SP_w}(s',t') \rVert} $ be the maximum ratio attained when $\mathit{SP_w}(s',t')$ bounds a polygon of type~$ P_1^1 $. Consider the corresponding~$ P_1^1$-triple between $ s$ and $ t$. Let $ H_3 $ be the cell for which the maximum ratio~$ \frac{\lVert \mathit{SGP_w}(u_j,u_{j+1}) \rVert}{\lVert \mathit{SP_w}(u_j,u_{j+1}) \rVert} $ is attained, where $ u_j, u_{j+1} \in H_3 $ are two consecutive points where $ \mathit{SP_{w}}(s,t) $ and $ X(s,t) $ coincide. Then, $ \lVert X(s,t) \rVert = \lVert \Pi_3^2(s,t) \rVert $.
		\end{lemma}
		
		\begin{lemma}
			\label{lem:equallengthhex122}
			Let $ X(s,t) \in G_{12\text{corner}} $. Let $ \frac{\lVert \mathit{SGP_w}(s',t') \rVert}{\lVert \mathit{SP_w}(s',t') \rVert} $ be the maximum ratio attained when $ \mathit{SP_w}(s',t') $ bounds a polygon of type~$ P_2^1 $. Consider the corresponding~$ P_2^1$-triple between $ s $ and $ t$. Let $ H_3 $ be the cell for which the maximum ratio~$ \frac{\lVert \mathit{SGP_w}(u_j,u_{j+1}) \rVert}{\lVert \mathit{SP_w}(u_j,u_{j+1}) \rVert} $ is attained, where $ u_j, u_{j+1} \in H_3 $ are two consecutive points where $ \mathit{SP_{w}}(s,t) $ and $ X(s,t) $ coincide. Then, $ \lVert X(s,t) \rVert = \lVert \Pi_3^2(s,t) \rVert = \lVert \Pi_3^3(s,t) \rVert $.
		\end{lemma}
        
        Now, we have all the tools needed to obtain upper bounds on the ratio~$ \frac{\lVert X(u_j, u_{j+1})\rVert}{\lVert \mathit{SP_w}(u_j, u_{j+1}) \rVert} $ for $X(u_j, u_{j+1})$ and $\mathit{SP_w}(u_j, u_{j+1}) $ bounding polygons of type~$ P^1_1 $ and $ P^1_2 $ when~$ X(s,t) \in G_{12\text{corner}} $. The proofs of the bounds in Lemmas~\ref{lem:ratio12corner1} and \ref{lem:ratio12corner2} require multiple case analysis. However, the analysis is similar to the proof of Lemma~\ref{lem:19} for square cells.

		\begin{lemma}
			\label{lem:ratio12corner1}
			Let $ u_j, \ u_{j+1} \in \mathcal{H} $ be two consecutive points where a shortest path~$ \mathit{SP_w}(s,t) $ and the crossing path $ X(s,t) \in G_{12\text{corner}} $ coincide. If $ u_j $ and $ u_{j+1} $ induce a polygon of type~$ P_1^1$, and $ \lVert X(s,t) \rVert = \lVert \Pi_3^2(s,t) \rVert $, then $ \frac{\lVert X(u_j, u_{j+1})\rVert}{\lVert \mathit{SP_w}(u_j, u_{j+1}) \rVert} \leq \frac{2}{\sqrt{2+\sqrt{3}}} \approx 1.04 $.
		\end{lemma}
		
		\begin{lemma}
			\label{lem:ratio12corner2}
			Let $ u_j, \ u_{j+1} \in \mathcal{H} $ be two consecutive points where a shortest path~$ \mathit{SP_w}(s,t) $ and the crossing path $ X(s,t) \in G_{12\text{corner}} $ coincide. If $ u_j $ and $ u_{j+1} $ induce a polygon of type~$ P_2^1$, and $ \lVert X(s,t) \rVert = \lVert \Pi_3^2(s,t) \rVert = \lVert \Pi_3^3(s,t) \rVert $, then $ \frac{\lVert X(u_j, u_{j+1})\rVert}{\lVert \mathit{SP_w}(u_j, u_{j+1}) \rVert} \leq \frac{2}{\sqrt{2+\sqrt{3}}} $.
		\end{lemma}

		Analogously as for Theorem~\ref{thm:ratiohex3corner}, and using Lemmas \ref{lem:ratio12corner1} and \ref{lem:ratio12corner2}, we obtain the following result.
		
		\begin{theorem}
		    \label{thm:ratiohex12corner}
		    In $ G_{12\text{corner}} $, $ \frac{\lVert \mathit{SGP_w}(s, t)\rVert}{\lVert \mathit{SP_w}(s,t) \rVert} \leq \frac{2}{\sqrt{2+\sqrt{3}}} $.
    	\end{theorem}
    	
    	Theorems \ref{thm:ratiohex3corner} and \ref{thm:ratiohex12corner} imply Corollaries \ref{cor:ratiohex3corner} and \ref{cor:ratiohex12corner}, respectively. These results are the upper bounds for the ratios involving the shortest vertex path $ \mathit{SVP_w}(s,t) $. The proofs are straightforward using that $ \lVert \mathit{SVP_w}(s,t) \rVert \geq \lVert \mathit{SP_w(s,t)} \rVert $.
    	
    	\begin{corollary}
		\label{cor:ratiohex3corner}
		In $ G_{3\text{corner}} $, $ \frac{\lVert \mathit{SGP_w}(s, t)\rVert}{\lVert \mathit{SVP_w}(s,t) \rVert} \leq \frac{3}{2} $.
	\end{corollary}
	
	\begin{corollary}
		\label{cor:ratiohex12corner}
		In $ G_{12\text{corner}} $, $ \frac{\lVert \mathit{SGP_w}(s, t)\rVert}{\lVert \mathit{SVP_w}(s,t) \rVert} \leq \frac{2}{\sqrt{2+\sqrt{3}}} $.
	\end{corollary}
	
	\subsection{Ratio $ \frac{\lVert \mathit{SVP_w}(s,t)\rVert}{\lVert \mathit{SP_w}(s,t)\rVert} $ for hexagonal cells}\label{sec:svphex}
	
	In this section we present two bounds on the ratio where the weighted shortest paths $ \mathit{SP_w}(s,t) $ and $ \mathit{SVP_w}(s,t) $ do not use any of the graphs $ G_{3\text{corner}} $ or $ G_{12\text{corner}} $, but $ G_{\text{corner}} $, i.e., the ratio $ \frac{\lVert \mathit{SVP_w}(s,t)\rVert}{\lVert \mathit{SP_w}(s,t)\rVert} $. As already mentioned in Section~\ref{sec:svpsquares}, the length of a weighted shortest vertex path $ \mathit{SVP_w}(s,t) $ is a lower bound for the length of a weighted shortest grid path $ \mathit{SGP_w}(s,t) $, so the upper bound on the ratio $ \frac{\lVert \mathit{SVP_w}(s, t)\rVert}{\lVert \mathit{SP_w}(s,t) \rVert} $ is obtained in Corollary~\ref{cor:7}.
	
	\begin{corollary}
		\label{cor:7}
		In a hexagonal tessellation, $ \frac{\lVert \mathit{SVP_w}(s, t)\rVert}{\lVert \mathit{SP_w}(s,t) \rVert} \leq \frac{2}{\sqrt{2+\sqrt{3}}} \approx 1.04 $.
	\end{corollary}
	
	Finally, we provide a lower bound for the ratio $ \frac{\lVert \mathit{SVP_w}(s, t)\rVert}{\lVert \mathit{SP_w}(s,t) \rVert} $. The green path in Figure~\ref{fig:lowerhex} is a weighted shortest vertex path~$ \mathit{SVP_w}(s,t) $ between vertices $ s $ and $ t $, thus, we have the following result.
	
	\begin{observation}
	    \label{obs:lowerhex}
	    In $ G_{\text{corner}} $, $ \frac{\lVert \mathit{SVP_w}(s, t)\rVert}{\lVert \mathit{SP_w}(s,t) \rVert} \geq \frac{2\sqrt{4\sqrt{3}-6}}{(2-\sqrt{3})(\sqrt{4\sqrt{3}-6}+6)} \approx 1.03 $.
	\end{observation}
	
	\begin{figure}[tb]
	    \centering
	    \includegraphics{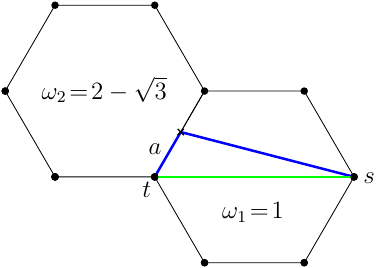}
	    \caption{$ \mathit{SP_w}(s,t) $ and $ \mathit{SVP_w}(s,t) $ are depicted in blue and green, respectively. The ratio~$ \frac{\lVert \mathit{SVP_w}(s,t)\rVert}{\lVert \mathit{SP_w}(s,t)\rVert} $ is $ \frac{2\sqrt{4\sqrt{3}-6}}{(2-\sqrt{3})(\sqrt{4\sqrt{3}-6}+6)} $ when $ a = 1-\sqrt{\frac{3(7-4\sqrt{3})}{4\sqrt{3}-6}}\approx 0.52 $ and a side length of mesh cell is $ 1$.}
	    \label{fig:lowerhex}
    \end{figure}

	\section{Conclusions}
            In this paper, we determined how the ratios between the weighted lengths of a shortest path, a shortest vertex path and a shortest grid path vary in a square and a hexagonal mesh. We studied the ratios for four different definitions of the neighbors of a vertex placed at a corner of a cell. This analysis is an extension of a previous work on triangular tessellations, see~\cite{bose2023approximating}. Thus, with this work we close the problem of calculating the ratios on every regular meshes that tessellates the continuous space. We believe that the techniques we propose in this paper can also be used to upper-bound the same ratios in 3D environments.
            
            As future work, it will be interesting to close the gap on the ratio $ \frac{\lVert \mathit{SVP_w}(s,t)\rVert}{\lVert \mathit{SP_w}(s,t)\rVert} $, although given how small it is, this would be only of theoretical interest. Another open problem to consider is to bound the ratios when the vertices are placed at the center of the cells. We showed in~\cite{bose2023approximating} that the ratios in this model in a triangular tessellation are unbounded. If a square mesh is considered, and every vertex is connected to four other vertices, the ratios involving $ \mathit{SGP_w}(s,t) $ are also unbounded, see Figure~\ref{fig:infinite}. However, for the case where each vertex is connected to eight neighboring vertices, or in a hexagonal mesh, it would be interesting to either improve the ratios that have been previously studied, see~\cite{Bound3}, or to conduct a study for the cases where no upper bounds have been obtained.

            \vspace{15mm}

{\small \noindent \textbf{Acknowledgments}}

P. B. is partially supported by NSERC. G. E., D. O. and R. I. S. are partially supported by project PID2019-104129GB-I00 funded by MICIU/AEI/10.13039/501100011033. G. E. is also funded by an FPU of the Universidad de Alcal\'a.

\bibliographystyle{abbrv}
\bibliography{arXiv-squares-v1}

\end{document}